\documentclass[12pt]{article}
\usepackage{fullpage}

\usepackage{algorithm, algorithmic}
\usepackage{authblk}
\usepackage{amsthm}
\usepackage{amsmath}
\usepackage{xcolor}
\usepackage{multirow}
\usepackage{wrapfig}
\usepackage[leftcaption]{sidecap}

\usepackage{thm-restate}

\usepackage{comment,amsfonts,amsmath,amsthm,graphicx, algorithm, algorithmic}

\definecolor{Darkblue}{rgb}{0,0,0.4}
\definecolor{Brown}{cmyk}{0,0.61,1.,0.60}
\definecolor{Purple}{cmyk}{0.45,0.86,0,0}
\definecolor{Darkgreen}{rgb}{0.133,0.543,0.133}

\usepackage[colorlinks,linkcolor=Darkblue,filecolor=blue,citecolor=blue,urlcolor=Darkblue,pagebackref]{hyperref}
\usepackage[nameinlink]{cleveref}

\newcommand{\commentout}[1]{}

\def\epsilon{\varepsilon}
\def\eps{\varepsilon}

\newcommand{\alert}[1]{\textbf{\color{red}
		[[[#1]]]}\marginpar{\textbf{\color{red}**}}\typeout{ALERT:\@
		\the\inputlineno: #1}}

\usepackage[colorinlistoftodos,prependcaption,textsize=tiny]{todonotes}

\def\eps{\epsilon}

\newcommand{\polylog}{{\rm polylog}}

\newcommand{\lca}{{\rm lca}}

\newcommand{\poly}{{\rm poly}}

\newcommand{\diam}{{\rm diam}}

\newcommand{\N}{\mathbb{N}}
\newcommand{\PP}{\mathcal{P}}

\newcommand{\HH}{\mathcal{H}}

\newcommand{\etal}{\emph{et.\ al. }}

\newcommand{\E}{{\mathbb{E}}}

\newcommand{\mommit}[1]{}
\newcommand{\namedref}[2]{\hyperref[#2]{#1~\ref*{#2}}}

\newcommand{\appendixref}[1]{\namedref{Appendix}{#1}}

\newcommand{\thmref}[1]{\namedref{Thm.}{#1}}
\newcommand{\corref}[1]{\namedref{Cor.}{#1}}

\theoremstyle{plain}
\newtheorem{theorem}{Theorem}[]
\newtheorem{lemma}{Lemma}[]
\newtheorem{claim}[lemma]{Claim}

\newtheorem{definition}[lemma]{Definition}
\newtheorem{corollary}[theorem]{Corollary}

\newcommand{\ddim}{{\rm ddim}}
\newcommand{\sddim}{{\rm ddim}}
\newcommand{\ba}{{\rm ba}}

\newcommand{\cD}{\mathcal{D}}

\newcommand{\cF}{\mathcal{F}}

\newcommand{\cJ}{\mathcal{J}}

\newcommand{\cP}{\mathcal{P}}
\newcommand{\cQ}{\mathcal{Q}}

\newcommand{\cT}{\mathcal{T}}

\newcommand{\SPDdepth}{\textsf{SPDdepth}}
\newcommand{\MST}{\mathrm{MST}}

\newcommand{\supp}{{\rm supp}}

\newcommand{\nonblind}[1]{}

\author[1]{Arnold Filtser\thanks{This research was supported by the Israel Science Foundation (grant No. 1042/22).}}
\author[2]{Yuval Gitlitz\thanks{This research was supported by the Israel Science Foundation (grant No. 970/21).}}
\author[2]{Ofer Neiman\thanks{This research was supported by the Israel Science Foundation (grant No. 970/21).}}

\affil[1]{Bar-Ilan University. Email: \texttt{arnold.filtser@biu.ac.il}}
\affil[2]{Ben-Gurion University of the Negev. Emails:  \texttt{gitlitz@post.bgu.ac.il}, \texttt{neimano@cs.bgu.ac.il}}

\title{Light, Reliable Spanners}

\begin{document}
\pagenumbering{gobble}

\maketitle

\begin{abstract}
    A \emph{$\nu$-reliable spanner} of a metric space $(X,d)$, is a (dominating) graph $H$, such that for any possible failure set $B\subseteq X$, there is a set $B^+$ just slightly larger $|B^+|\le(1+\nu)\cdot|B|$, and all distances between pairs in $X\setminus B^+$ are (approximately) preserved in $H\setminus B$. Recently, there have been several works on sparse reliable spanners in various settings, but so far, the weight of such spanners has not been analyzed at all. In this work, we initiate the study of \emph{light} reliable spanners, whose weight is proportional to that of the Minimum Spanning Tree (MST) of $X$. 
    
    We first observe that unlike sparsity, the lightness of any 
deterministic reliable spanner is huge, even for the metric of the simple path graph. Therefore, randomness must be used: an \emph{oblivious} reliable spanner is a distribution over spanners, and the bound on $|B^+|$ holds in expectation.

We devise an oblivious $\nu$-reliable $(2+\frac{2}{k-1})$-spanner for any $k$-HST, whose lightness is $\approx \nu^{-2}$. We demonstrate a matching $\Omega(\nu^{-2})$ lower bound on the lightness (for any finite stretch). We also note that any stretch below 2 must incur linear lightness.

For general metrics, doubling metrics, and metrics arising from minor-free graphs, we construct {\em light} tree covers, in which every tree is a $k$-HST of low weight. Combining these covers with our results for $k$-HSTs, we obtain oblivious reliable light spanners for these metric spaces, with nearly optimal parameters. In particular, for doubling metrics we get an oblivious $\nu$-reliable $(1+\varepsilon)$-spanner with lightness $\varepsilon^{-O(\sddim)}\cdot\tilde{O}(\nu^{-2}\cdot\log n)$, which is best possible (up to lower order terms).

    \end{abstract}
\newpage
\addtocontents{toc}{\protect\setcounter{tocdepth}{2}}
\tableofcontents
\newpage
\pagenumbering{arabic}

\section{Introduction}
Given a metric space $(X,d_X)$, a $t$-{\em spanner} is a graph $H$ over $X$ such that for every $x,y\in X$, $d_X(x,y)\le d_H(x,y)\le t\cdot d_X(x,y)$, where $d_H$ is the shortest path metric in $H$. \footnote{Often in the literature, the input metric is the shortest path metric of a graph, and a spanner is
required to be a subgraph of the input graph. Here we study metric spanners where there is no such requirement.}
The parameter $t$ is often referred to as the \emph{stretch}.
In essence, the purpose of spanners is to represent the distance metric using a sparse graph. Spanners where introduced by Peleg and Sch\"{a}ffer \cite{PS89}, and found numerous applications throughout computer science. For a more systematical study, we refer to the book \cite{NS07} and survey \cite{ABSHJKS20}.
In many cases, the goal is to minimize the total weight of the spanner and not just the number of edges. E.g., when constructing a road network,
the cost is better measured by the total length of paved roads, as opposed to their number. 
This parameter of interest is formalized as the \emph{lightness} of a spanner, which is the ratio between the weight of the spanner (sum of all edge weights), and the weight of the  Minimum Spanning Tree (MST) of $X$: $\frac{w(H)}{w(\rm{MST})}$. Note that the MST is the minimal weight of a connected graph, and thus of a spanner with finite stretch. So the lightness is simply a ``normalized'' notion of weight.

Light spanners have been thoroughly studied. It is known that general $n$-point metric spaces admit a $(2k-1)(1+\eps)$ spanner (for $k\in\N$, $\eps\in(0,1)$) with $O(n^{1+1/k})$ edges and lightness $O(\eps^{-1}\cdot n^{1/k})$
\cite{LS23,Bodwin23} (see also \cite{ADDJS93,ENS15,CW18,FS20}). Every $n$-point metric space with doubling dimension\footnote{A metric space $(X, d)$ has doubling dimension $\ddim$ if every ball of radius $2r$ can be 	covered by $2^{\sddim}$ balls of radius $r$. The $d$-dimensional Euclidean space has doubling dimension $\Theta(d)$.\label{foot:doubling}} $\ddim$ admits a $(1+\eps)$-spanner with $n\cdot\eps^{-O(\sddim)}$ edges and lightness $\eps^{-O(\sddim)}$ \cite{BLW19} (see also \cite{Gottlieb15,FS20}).
Finally, the shortest path metric of a graph excluding a fixed minor admits a (sub-graph, which already implies sparsity) $(1+\eps)$-spanner with lightness $\tilde{O}(\eps^{-3})$ \cite{BLW17}.

A highly desirable properly of a spanner is the ability to withstand massive node-failures. 
To this end, Bose \etal \cite{BDMS13} introduced the notion of a \emph{reliable spanner}. \footnote{For a comprehensive discussion with the related notion of fault-tolerant spanners, see \Cref{subsec:related}.}
 Here, given a set of failed nodes $B\subseteq X$, the residual spanner $H\setminus B$ is a $t$-spanner for $X\setminus B^+$, where $B^+\supseteq B$ is a set slightly larger than $B$.
 For the case of points in $d$-dimensional Euclidean space, for constant $d$, Bose \etal \cite{BDMS13} constructed $O(1)$ spanner such that $|B^+|\le O(|B|^2)$.
 Later, Buchin, Har-Peled, and Ol{\'{a}}h \cite{BHO19} constructed $1+\eps$ reliable spanner with $n\cdot\eps^{-O(d)}\cdot\nu^{-6}\cdot\tilde{O}(\log n)$ edges, guaranteeing that for every set of failed nodes $B$, $|B^+|\le (1+\nu)\cdot|B|$.
 This result was generalized to metric spaces with doubling dimension $\ddim$ by Filtser and Le \cite{FL22}.
 
 While reliable spanners for Euclidean and doubling metrics admit sparsity which is comparable to their non-reliable counter-parts, the situation is very different for other metric families. Indeed, Har-Peled \etal \cite{HMO21} showed that every reliable $k$-spanner of the simple uniform metric (which is also a tree metric) must have $\Omega(n^{1+1/k})$ edges.
 Nevertheless, it is possible to construct \emph{oblivious} reliable spanner for other metric spaces with good parameters, where the bound on the size of $B^+$ is only in expectation.

\begin{definition}[Reliable spanner]\label{def:reliableSpanner}
	A weighted graph $H$ over point set $X$ is a deterministic $\nu$-reliable $t$-spanner
	of a metric space $(X,d_{X})$ if $d_{H}$ dominates\footnote{Metric space $(X,d_H)$ dominates metric space $(X,d_X)$ if  $\forall u,v\in X$, $d_X(u,v)\le d_H(u,v)$.\label{foot:dominating}} 
	$d_{X}$, and for every
	set $B\subseteq X$ of points, called an \emph{attack set}, there is a set $B^{+}\supseteq B$, called a \emph{faulty extension} of $B$, 	such that:
 		(1) $|B^{+}|\le(1+\nu)|B|$.
		(2) For every $x,y\notin B^{+}$, $d_{H[X\setminus B]}(x,y)\le t\cdot d_{X}(x,y)$.

	An oblivious $\nu$-reliable $t$-spanner is a distribution $\mathcal{D}$ over dominating graphs $H$, such that for every attack set $B\subseteq X$ and $H\in\supp(\mathcal{D})$,
	there 
	exist a superset $B_H^{+}\supseteq B$ such that, for
	every $x,y\notin B_H^{+}$, $d_{H[X\setminus B]}(x,y)\le t\cdot d_{X}(x,y)$,
	and $\mathbb{E}_{H\sim\mathcal{D}}\left[|B_H^{+}|\right]\le(1+\nu)|B|$. We say that the oblivious spanner $\mathcal{D}$ has $m$ edges and lightness $\phi$ if every $H\in\supp(\mathcal{D})$ has at most $m$ edges and lightness at most $\phi$. 

\end{definition}

For general $n$-point metrics, Filtser and Le \cite{FL22} (improving over \cite{HMO21}) constructed an oblivious $\nu$-reliable $8k+\eps$-spanner with $\tilde{O}(n^{1+\frac1k}\cdot\eps^{-2})\cdot\nu^{-1}$ edges.
For the shortest path metric of graph excluding a fixed minor, there is oblivious $\nu$-reliable ($2+\eps$)-spanner with $\eps^{-2}\cdot\nu^{-1}\cdot \tilde{O}(n)$ edges, while every oblivious reliable spanner with stretch $t<2$ requires $\Omega(n^2)$ edges \cite{FL22}.
For Euclidean and doubling metrics, oblivious $\nu$-reliable $(1+\eps)$-spanners can be constructed with only $n\cdot\eps^{-O(d)}\cdot\tilde{O}(\nu^{-1}\cdot\log^2\log n)$ edges \cite{BHO20,FL22}.

But what about lightness? no previous work attempted to construct reliable spanners of low total weight even though it is clearly desirable to construct reliable networks of low total cost. 
The single most studied metric in the context of reliable spanners is the unweighted path $P_n$. Indeed, most of the previous work \cite{BHO19,BHO20,FL22,Fil23} focused on constructing various reliable $1$-spanners for the path graph, and then generalized it other metric spaces using \emph{locality sensitive orderings} \footnote{Locality sensitive ordering is a generic tool that ``reduces'' metric spaces into the line, by devising a collection of orderings such that every two points are ``nearby'' in one of the orderings, see \cite{CHJ20,FL22}.}.
A reliable spanner should have many edges between every two large enough sets, so that they could not be easily disconnected. Consider an attack $B$ consisting of the middle $\frac n2$ vertices on $P_n$. If there are less than $\frac n8$ crossing edges from left to right, then an attack $B'\supseteq B$ that contains also one endpoint per crossing edge, will disconnect two sets of size $\frac n8$. Therefore a linear number of vertices should be added to $B'^+$. We conclude that every deterministic reliable spanner (for any finite stretch) must have lightness $\Omega(n)$ (see \Cref{thm:PathDeterministicLB} for a formal proof). 
Thus, all hope lies in oblivious reliable spanners. 
However, even here any two large sets must be well connected. 
Previous oblivious reliable spanners for $P_n$ all had unacceptable polynomial lightness. 

As reliable spanners for $P_n$ are the main building blocks for reliable spanners for other metric spaces, all previous constructions have inherent polynomial lightness.\footnote{The only previous work that did not reduced to $P_n$ is by Har-Peled \etal \cite{HMO21} who reduced to uniform metrics. Nevertheless, their approach on $P_n$ will have stretch $3$, and lightness $\Omega(n)$.}

\begin{table}[]
	\begin{center}
		\def\arraystretch{1.15}
		\begin{tabular}{|l|l|l|l|l|}
			\hline
			Family & Stretch   & Lightness & Size & Ref\\ \hline
			\multicolumn{1}{|c|}{\multirow{2}{*}{Doubling $\ddim$ }}     & $1+\eps$  & $\eps^{-O(\sddim)}\cdot\tilde{O}(\nu^{-2}\cdot\log n)$          & $n\cdot\eps^{-O(\sddim)}\cdot\tilde{O}(\nu^{-2})\cdot*$     &  \corref{thm:doubling}   \\ \cline{2-5} 
			& $\ddim$ &   $\tilde{O}(\log n\cdot\nu^{-2})\cdot\ddim^{O(1)}$        &  $n\cdot\tilde{O}\left(\nu^{-2}\right)\cdot\ddim^{O(1)}\cdot*$    & \corref{thm:DdimLarge}  \\ \hline
			
			\multicolumn{1}{|c|}{\multirow{2}{*}{General Metric }}   & $12t+\eps$ &   $n^{1/t}\cdot\tilde{O}(\nu^{-2}\cdot\eps^{-4})\cdot\log^{O(1)}n$        &  $\tilde{O}\left(n^{1+1/t}\cdot\nu^{-2}\cdot\eps^{-3}\right)$    &   \corref{thm:general}\\ \cline{2-5} 
			
			& $O(\log n)$ &  $\tilde{O}(\nu^{-2}\cdot\log^4 n)$         &  $n\cdot \tilde{O}\left(\nu^{-2}\cdot\log^{3}n\right)$     &  \corref{cor:HSTcoverGeneralMetricLogn} \\ \hline
			Minor-Free    & $2+\eps$  &   $\tilde{O}(\nu^{-2}\cdot\eps^{-7}\cdot\log^8 n)$         &  $\tilde{O}(n\cdot \nu^{-2}\cdot\eps^{-6})$    & \thmref{thm:MinorFreeOptimalStretch}\\ \hline
			Tree		& $<2$   & $\Omega(n)$ & $\Omega(n^2)$ & \cite{FL22}\\\hline
			Weighted Path & $1$ &  $\nu^{-2}\cdot\tilde{O}(\log n)$         & $n\cdot\tilde{O}(\nu^{-1})\cdot *$     &\corref{cor:path:logn}   \\ \hline
			\multicolumn{1}{|c|}{\multirow{2}{*}{\begin{tabular}[c]{@{}c@{}}Unweighted \\ Path\end{tabular}}} & $<\infty$       & $\Omega(\nu^{-2}\cdot\log(\nu\cdot n))$          &  -    & \thmref{thm:PathObliviousLB}    \\ \cline{2-5} 
			& $<\infty$ &  $\Omega(n)$ (deterministic)        &   -   & \thmref{thm:PathDeterministicLB}     \\ \hline
			
			\multicolumn{1}{|c|}{\multirow{2}{*}{\begin{tabular}[c]{@{}c@{}}HST\\ (ultrametric)\end{tabular}}}   & $2+\eps$   
			&  $\tilde{O}(\eps^{-4}\cdot\nu^{-2})\cdot *$  &  $n\cdot\tilde{O}\left(\eps^{-3}\cdot\nu^{-2}\right)\cdot*$  &  \thmref{thm:GeneralUltrametric}   \\ \cline{2-5} 
			& $<\infty$ &  $\Omega(\nu^{-2})$  & - &  \thmref{thm:HSTObliviousLB}   \\ \hline                                                                   
		\end{tabular}
		\caption{\small Our results for constructing light $\nu$-reliable spanners for various metric spaces. All the results in the table (other than the one specified as deterministic) are for oblivious reliable spanners. 
			Stretch $<\infty$ stands for the requirement that all the points in $X\setminus B^+$ belong to the same connected component in $H\setminus B$.
			$*$ stands for $\poly(\log\log n)$ factors.}\label{tab:results}
	\end{center}
\end{table}

\subsection{Our Results}\label{subsec:results}
The results of this paper are summarized in \Cref{tab:results}. Our results on light reliable spanners for various metric families are based on constructing such spanners for $k$-HSTs, this lies in contrast to previous results on sparse reliable spanners, which were mostly based on reliable spanners for the path graph. 

Roughly speaking, previous works on reliable spanners show us that the ``cost'' of making a spanner $\nu$-reliable, is often a $\nu^{-1}$ factor in its size. 
Our results in this paper offer a similar view for light spanners: here the ``cost'' of reliability is a factor of $\nu^{-2}$ in the lightness. 
That is, an $\Omega(\nu^{-2})$ factor must be paid in the most basic cases (path graph, HST), while in more interesting and complicated metric families, we essentially match the best non-reliable light spanner constructions, up to this $\nu^{-2}$ factor (and in some cases, such as minor-free graphs, an unavoidable constant increase in the stretch). 
For brevity, in the discussion that follows we omit the bounds on the size of our spanners (which can be found in \Cref{tab:results}).

\paragraph{$k$-HSTs.} We devise an oblivious $\nu$-reliable $2+\frac{O(1)}{k}$-spanner for any $k$-HST (see \Cref{def:HST}), whose lightness is $\tilde{O}(\nu^{-1}\cdot\log\log n)^2$ (see \Cref{thm:k-HST-new}). It is implicitly shown in \cite[Observation 1]{FL22} that with stretch smaller than 2, the lightness must be $\Omega(n)$. So when $k$ is large, our stretch bound is nearly optimal.\footnote{We also have a similar result for every $k\ge 1$, with stretch $2+\varepsilon$ and lightness $\tilde{O}(\varepsilon^2\cdot\nu^{-1}\cdot\log\log n)^2$.} We also show that the lightness must be at least $\Omega(\nu^{-2})$, regardless of the stretch, thus nearly matching our upper bound.

\paragraph{Light $k$-HST Covers.} To obtain additional results for other metric families, following \cite{FL22}, we use the notion of {\em tree covers}, in which every tree is a $k$-HST (see \Cref{def:hst-cover}). We design these covers for metrics admitting a pairwise partition cover scheme (see \Cref{def:ppcs}), such that each $k$-HST in the cover has lightness $O(k\cdot\log n)$.

\paragraph{General Metrics.} For any metric space, by building a light $k$-HST cover, and applying our oblivious reliable spanner for every $k$-HST in the cover, we obtain an oblivious $\nu$-reliable $O(k)$-spanner with lightness $\tilde{O}(\nu^{-2}\cdot n^{1/k})$. Note that up to a constant in the stretch (and lower order terms), this result is optimal, even omitting the reliability requirement.

\paragraph{Doubling Metrics.} For any metric with doubling dimension $\ddim$,$^{\ref{foot:doubling}}$ and $\eps\in(0,1)$, we devise an oblivious $\nu$-reliable $(1+\eps)$-spanner with lightness $\eps^{-O(\sddim)}\cdot\tilde{O}\left(\nu^{-2}\cdot\log n\right)$. 
This result is tight 
up to second order terms. 
Indeed, it is folklore that any $(1+\eps)$-spanner for doubling metrics must have lightness $\eps^{-\Omega(\sddim)}$ (see e.g., \cite{BLW19}). In \Cref{thm:PathObliviousLB}, we show that every oblivious $\nu$-reliable spanner (for any finite stretch) for the shortest path metric of the unweighted path graph (which has $\ddim$ $1$) must have lightness $\Omega(\nu^{-2}\cdot\log(\nu n))$.
This dependence on $n$ in the lower bound is somewhat surprising, and does not appear in the closely related fault-tolerant spanners for doubling metrics (see \Cref{subsec:related} for further details).

In our doubling reliable spanner construction, we adapt the framework used for general metrics. Note that general $k$-HSTs must suffer stretch at least 2. 
Fortunately, the $k$-HSTs in the cover for doubling metrics have bounded maximum degree. For such HSTs we construct oblivious reliable $1+O(\frac{1}{k})$-spanner with lightness $\tilde{O}(\nu^{-1}\cdot \log\log n)^2$ (see \Cref{thm:k-HST-degree-new}). Whenever $k\ge\frac1\eps$, this is $1+\eps$ stretch.

\paragraph{Metrics of Minor-free Graphs.}
Consider a metric $(X,d)$ arising from shortest paths of a graph $G$ that excludes a fixed minor.
In \Cref{thm:minor_free:pairwise} we show that $X$ admits ``good'' pairwise partition cover with stretch $2$, and thus by using the framework mentioned above as a black-box, we can get oblivious $\nu$-reliable $(4+\epsilon)$-spanner (see \Cref{thm:MinorFreeStretch4}). However, the lower bound on the stretch is the same as for $k$-HST, which is only 2 (whenever the lightness is sub-linear). To obtain near optimal results, we exploit a certain property of our pairwise partition cover for these metrics, and achieve (in a non-black-box manner) the nearly optimal oblivious $\nu$-reliable $(2+\varepsilon)$-spanner with lightness $\nu^{-2}\cdot\poly(\log n,1/\varepsilon)$.

\paragraph{The path graph.} We conclude our journey on light reliable spanners by constructing an oblivious $\nu$-reliable 1-spanner for the weighted path graph $P_n$, whose lightness is $\tilde{O}(\nu^{-2}\cdot\log n)$. As mentioned above, we prove that this bound on the lightness is optimal (up to lower order terms), for any finite stretch. 
A useful\footnote{Buchin \etal's \cite{BHO20} oblivious reliable spanner for the path is $O(\log n)$ hop-bounded. This property was crucial for the construction of sparse oblivious reliable spanners for Euclidean and doubling metrics \cite{BHO20,FL22}. Filtser and Le \cite{FL22} constructed $2$ hop-bounded oblivious reliable spanners for the path (and used it in their construction of oblivious reliable spanner for general metrics)} 
property of our spanner is that it is hop-bounded,  that is, every pair outside $B^+$ admits a shortest path with at most $\log n$ edges.\footnote{Additionally, for any $h\ge 1$, we can also devise a $h$-hop-bounded reliable spanner, while achieving lightness $\approx\nu^{-2}\cdot h\cdot n^{1/h}$, see \Cref{thm:path}.} 

\subsection{Technical Overview}

From a high level, our construction of light reliable spanners for various graph families has the following structure. 
\begin{itemize}
    \item We first devise light reliable spanners for $k$-HSTs.
    \item We construct {\em light} tree covers for the relevant family, where all the trees in the cover are $k$-HSTs. 
    \item The final step is to sample a reliable spanner for each tree in the cover, and take as a final spanner the union of these spanners.
\end{itemize}  

In what follows we elaborate more on the main ideas and techniques for each of those steps.

\subsubsection{Reliable Light Spanner for $k$-HSTs}

Let $T$ be the tree representing the $k$-HST (see \Cref{def:HST}). 
Our construction consists of a collection of randomly chosen bi-cliques: For every node $x\in T$ we choose at random a set $Z_x$ of $\ell\approx\nu^{-1}$ vertices from the leaves of the subtree rooted at $x$ (denoted $L(x)$). Then, for every $x\in T$ with children $x_1,\dots,x_t$, add to the spanner $H$ all edges in $Z_x\times Z_{x_j}$ for every
$j=1,\dots, t$.

Fix a pair of leaves $u,v\in T$, let $x=\lca(u,v)$, and let $x_i$ (resp., $x_j$) be the child of $x$ whose subtree contains $u$ (resp., $v$). The idea behind finding a spanner path between $u,v$ is as follows. We will connect both $u,v$ to a certain chosen leaf $x'\in Z_x$. To this end, we first connect recursively $u$ to a  $u'\in Z_{x_i}$ and $v$ to $v'\in Z_{x_j}$. Now, if $x$, $x_i$, and $x_j$ have all chosen such leaves $x',u',v'$ to the sets $Z_x,Z_{x_i},Z_{x_j}$ respectively, that survive the attack $B$, and also we managed the $u-u'$ and $v-v'$ connections recursively, then we can complete the $u-v$ path. That path will consists of the two ``long'' bi-clique edges $\{u',x'\},\{x',v'\}$, and the recursive $u-u'$ and $v-v'$ paths. Note that since $u,u'\in L(x_i)$, $d_T(u,u')\le d_T(u,v)/k$ (and similarly $d_T(v,v')\le d_T(u,v)/k$), so we can show inductively that the total distance taken by these recursive paths is only $O(d_T(u,v)/k)$.
See \Cref{fig:kHSTintution} for an illustration of a path in $H$ between two vertices $u,v$.

Having established what is needed for finding a spanner path, we say that a leaf is {\em safe} if all its ancestors $x$ in $T$ have that $Z_x$ is not fully included in $B$. The failure set $B^+$ consists of $B$ and all leaves that are not safe.

A subtle issue is that a vertex may have a linear number of ancestors, and we will need $\ell$ to be at least logarithmic to ensure good probability for success in all of them. To avoid this, 
we use the following approach. For any node $x$ that has a ``heavy''child $y$ (that is, $L(y)$ is almost as large as $L(x)$), we use the sample $Z_y$ for $x$, instead of sampling $Z_x$. This way, any leaf will have only logarithmically many ancestors that are not heavy parents, which reduce dramatically the sample size needed for success in all ancestors. 

For the reliability analysis, we distinguish between leaves that have an ancestor $x$ with a very large $1-\nu$ fraction of vertices in $L(x)$ that fall in the attack set $B$. These leaves are immediately taken as failed, but there can be only $\approx\nu|B|$ such leaves. For the other leaves, a delicate technical analysis follows to show that only a small fraction $\approx \nu\cdot|B|$ new vertices are expected to join $B^+$. Note that if some node has a heavy child, we take the child's sample, so some care is needed in the analysis to account for this -- roughly speaking, the definition of ``heavy'' must depend on the reliability parameter $\nu$, in order to ensure sufficiently small failure probability.

\paragraph{Improved stretch for bounded degree HSTs.} In the case the $k$-HST has bounded degree $\delta$, we can alter the construction slightly, and for every $x$ with children $x_1,\dots,x_s$, also add all edges in $Z_{x_i}\times Z_{x_j}$ for every $1\le i<j\le s$. While this alternative increases the lightness and size by a factor of $\delta$, the stretch improves to $1+\frac{O(1)}{k}$, since we only use one long edge. This variation will be useful for the class of doubling metrics.

\begin{SCfigure}[][t]\caption{\it 
		Illustration of the construction of spanner for a $k$-HST. For each internal node $x$ we sample a subset $Z_x$ of leaves from $L(x)$, and connect all of $Z_x$ to $Z_{x'}$ for every child $x'$ of $x$. The path from $u$ to $v$ will first go from $u$ to a surviving vertex in $Z_{x_i}$ (using recursion), from there to a surviving vertices in $Z_x$ and $Z_{x_j}$, and finally to $v$ (again by recursion).\\
		\label{fig:kHSTintution}}
	\includegraphics[width=.5\textwidth]{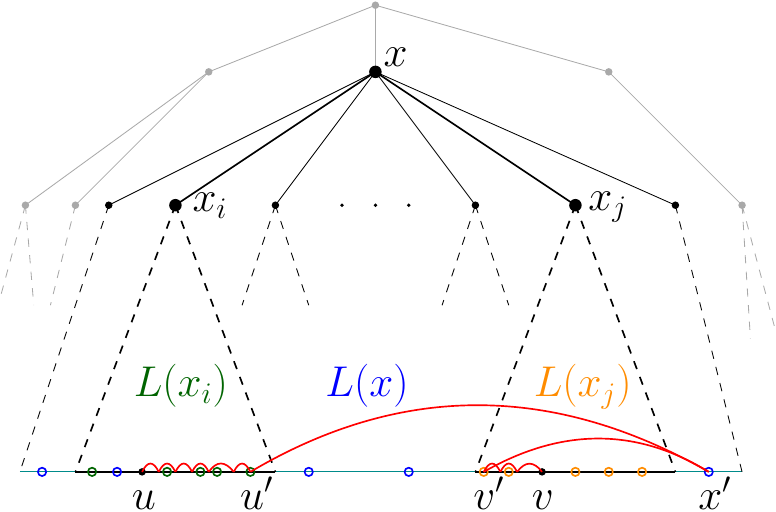}	
\end{SCfigure}

\subsubsection{Reliable Spanners via Light $k$-HST Covers}

A $(\tau,\rho)$-tree cover of a metric space $(X,d)$, is a collection of $\tau$ dominating trees, such that for every pair $u,v\in X$, there exists a tree $T$ in the cover with $d_T(u,v)\le \rho\cdot d(u,v)$. Let $(X,d)$ be any metric that admits a $(\tau,\rho)$-tree cover in which all trees are $k$-HSTs of weight at most $O(l\cdot w(MST(X))$, then we can devise an oblivious reliable spanner for $X$ as follows. Sample an oblivious light $\nu/\tau$-reliable spanner $H_T$ for each tree $T$, and define $H=\bigcup_T H_T$ as their union. We define $B^+$ as the union of all the failure sets $B^+_T$ over all tree spanners.

Since in every $\nu/\tau$-reliable spanner of a tree only $\nu/\tau\cdot |B|$ additional vertices fail in expectation, the total expected number of additional failures is at most $\nu\cdot|B|$, as required. Now, if a pair $u,v$ did not fail, there is a $k$-HST $T$ in which $d_T(u,v)\le \rho\cdot d(u,v)$, and thus $H$ has stretch at most $\rho\cdot (2+\frac{O(1)}{k})$ for such a pair.

\paragraph{ Light $k$-HST Covers using Pairwise Partition Cover Scheme.}

A $(\tau,\rho,\epsilon,\Delta)$-Pairwise Partition Cover for a metric space $(X,d)$ is a collection of $\tau$ partitions, each cluster in each partition has diameter at most $\Delta$, and every pair $u,v\in X$ with $\frac{\Delta}{2\rho}\le d(u,v)\le \frac{\Delta}{\rho}$ is {\em padded} in at least one cluster $C$ of a partition. This means that the cluster $C$ contains $u,v$, and also the balls of radius $\epsilon\Delta$ around them, see \Cref{def:ppcs}. If $(X,d)$ admits such a cover for every $\Delta$, we say it has a Pairwise Partition Cover Scheme (PPCS). In \cite{FL22}, PPCS were shown for general metrics and doubling metrics. In this paper, for any parameter $0<\epsilon<1/6$, we devise a $\left(\frac{\log n}{\epsilon},\frac{2}{1-6\epsilon},\epsilon\right)$-PPCS for minor-free graphs.

In \cite{FL22} it was shown that one can obtain a $k$-HST cover from a PPCS, in such a way that every cluster of diameter $\Delta$ in the PPCS corresponds to an internal node $x$ of one of the $k$-HSTs, with label $\Gamma_x=\Delta$.
For our purposes, we want every $k$-HST in the cover to be light. To this end, we augment the reduction of \cite{FL22} by a feature that allows us to bound the lightness of the resulting $k$-HST. The idea is to use {\em nets}, see \Cref{def:net}. A basic observation for a $\Delta$-net ${\cal N}$ of a metric space $(X,d)$, is that $w(MST(X))\ge \Omega(|{\cal N}|\cdot\Delta)$. On the other hand, the weight of a $k$-HST $T$ is roughly $\sum_{x\in T}k\cdot\Gamma_x$ (every node pays for the edge to its parent in $T$). So as long as the number of internal nodes with label $\Delta$ is bounded by $|{\cal N}|$, the $k$-HST will be rather light.

Now, given some partition with diameter bound $\Delta$, we take a $\approx\epsilon\Delta$-net ${\cal N}$, and break all clusters that do not contain a net point. Then the points in the broken clusters are joined to a nearby remaining cluster. Since the net is dense enough, each cluster that was used for padding remains intact, while the number of clusters is bounded by $|{\cal N}|$. This enables us to bound the weight of the $k$-HST accordingly.

\subsubsection{Reliable Light Spanner for Minor-free Graphs with $2+\eps$ stretch}

In the special case of minor-free graphs, the framework described above will lose a factor of 2 in the stretch in two places. The first is due to the padding of the PPCS, and the second in the reliable spanners for the $k$-HSTs. While each of these losses is unavoidable,\footnote{Stretch $2$ for HST is necessary: Consider the uniform metric, every spanner with less than ${n\choose2}$ edges has stretch $2$. Every PPCS for minor free graphs must have either $\rho\ge2$ or $\tau=\Omega(n)$: Fix $\rho<2$, and consider the unweighted star graph. There are $n-1$ leaf-center pairs, while a single partition can satisfy at most a single pair.}
we can still exploit a certain property of our PPCS for minor-free graphs, to improve the stretch to near optimal $2+\epsilon$. 

In our previous approach, suppose vertices $u,v$ are padded in some cluster $C$ of the PPCS, with diameter at most $\Delta$. Then in the $k$-HST cover, we will have some tree with an internal node $x$ corresponding to $C$, whose label is $\Gamma_x=\Delta$. The way we construct the spanner path between $u,v$ is via some chosen leaf $z$ in $L(x)$, and as both $d(u,z)$, $d(v,z)$ can be as large as $\Delta$, we loose a factor of 2 here.

The main observation behind overcoming this loss, is that in our PPCS for minor-free graphs, each cluster $C$ is a ball around some center $x$, and whenever a pair $u,v$ is padded, then $x$ is very close to the shortest $u-v$ path, meaning that $d(u,x)+d(v,x)\le(1+\epsilon)\cdot d(u,v)$. While we cannot guarantee that $x$, or a vertex close to $x$, will survive the attack $B$, we can still use this to improve the stretch guarantee.
Suppose that $Z_x$ contains a surviving leaf $z$ which is closer to $x$ than both $u,v$, then
\[
d(u,z)+d(z,v)\le(d(u,x)+d(x,z))+(d(z,x)+d(x,v))\le2(d(u,x)+d(x,v))\le2(1+\epsilon)\cdot d(u,v)~.
\]
So, instead of sampling a set $Z_x$ of leaves at random from $L(x)$, we create a bias towards vertices closer to the center $x$. Concretely, order the leaves of $L(x)$ by their distance to $x$, and we would like that the probability of the $j$-th leaf in $L(x)$ to join $Z_x$ will be $\approx\frac{1}{j}$. This way, the expected size of $Z_x$ is still small, and if not too many vertices in the appropriate prefix of $L(x)$ are in $B$, then there is a good probability that such a $z\in Z_x$ exists. However, as it turns out, this requirement it too strict, since every internal node $x$ will force us to move vertices to $B^+$ that fail due many vertices in $B$ in its induced ordering.

To avoid this hurdle, we use a {\em global} ordering for all internal nodes -- a carefully chosen preorder of $T$ -- and prove that the induced order on $L(x)$ is a good enough approximation of distances to $x$ (specifically, up to an additive factor of $\approx\Gamma_x/k$).

\subsubsection{Reliable Light Spanner for the Path Graph}
There were several construction of a reliable spanner for $P_n$ in previous works \cite{BHO19,BHO20,FL22}, none of them could provide a meaningful bound on the lightness. For instance, the first step in the construction of \cite{BHO19}
was to connect the first $n/2$ vertices to the last $n/2$ vertices via a bipartite expander graph. In particular, the total weight of just this step is $\Omega(n^2)$.
The method of \cite{FL22} is to sample $\approx\nu^{-1}$ vertices as star centers, and connect all other vertices to each center. This construction also clearly isn't light, as the total weight of even one such star is $\Omega(n^2)$.

Our construction of an oblivious light $\nu$-reliable spanner for (weighted) $P_n$ is similar to the approach taken by \cite{BHO20}. It starts by sampling a laminar collection of subsets $[n]=V_0\supseteq V_1\supseteq V_2\supseteq\cdots\supseteq V_{\log n}$, where $|V_i|$ contains $\frac{n}{2^i}$ points in expectation. 
However, the construction of \cite{BHO20} used long range edges: from vertices in $V_i$ to the nearest $\approx 2^{i/2}$ other vertices in $V_i$, and thus its lightness is polynomial in $n$.\footnote{To see why the lightness is polynomial, consider just the level $i=\frac23\log n$, then $|V_i|\approx n^{1/3}$, but also the number of connected neighbors is $2^{i/2}=n^{1/3}$, so all $\approx n^{2/3}$ edges between vertices in $V_i$ are added. The average length of these edges is linear in $n$, so the lightness is $\Omega(n^{2/3})$.}

To ensure bounded lightness, we take a more local approach, and each point $a\in V_i$ adds edges to only the nearest $\ell\approx\nu^{-1}$ points in $V_i$ and $V_{i+1}$ on both its left and right sides. 
We remark that the connections to the next level are crucial in order to avoid additional logarithmic factors (since unlike \cite{BHO20}, we cannot use the exponentially far away vertices, that would have provided high probability for connection of every vertex to the next level).
The lightness follows as each edge $e$ of $P$ is expected to be ``covered'' $\ell^2$ times, in each of the $\log n$ levels.

The reliability analysis of our spanner uses the notion of {\em shadow}, introduced by \cite{BHO19}. For the path $P_n$, roughly speaking, a vertex $u$ is outside the $\alpha$-shadow of an attack $B$, if in all intervals containing $u$, there is at most an $\alpha$ fraction of failed vertices (in $B$).

The reliability argument goes as follows: a vertex $a\in[n]\setminus B$ fails and joins $B^+$ only if there exists a level $i$ in which all its  connections to $V_{i+1}$ fail. That is, its $\ell$ closest vertices in $V_{i+1}$ are in $B$. But as points are chosen to $V_{i+1}$ independently of $B$, this is an unlikely event, whose probability can be bounded as a function of the largest $\alpha$-shadow that does not contain $a$. To obtain our tight bound, we need a delicate case-analysis for the different regimes of $\alpha$-shadows.

The stretch analysis is a refinement of \cite{BHO20} stairway approach. A nice feature is that each pair in $[n]\setminus B^+$ will have a shortest path of at most $\log n$ hops in the spanner $H$.
 
\subsection{Related Work}\label{subsec:related}
\paragraph{Light fault-tolerant spanners.}
Levcopoulos \etal \cite{LNS98} introduced the notion of $f$-fault-tolerant spanner, where it is guaranteed that for every set $F$ of at most $f$ faulty nodes, $H\setminus F$ is a $t$-spanner of $X\setminus F$.
However, the parameter $f$ has to be specified in advance, and both sparsity and lightness of the spanner must polynomially depend on $f$. Thus, unlike reliable spanners,
it is impossible to construct sparse and light fault-tolerant spanners that can withstand scenarios where, say, half of the nodes fail.

Czumaj and Zhao \cite{CZ04} constructed $f$ fault-tolerant spanners for point in constant dimensional Euclidean space with optimal $O(f^2)$ lightness (improving over \cite{LNS98} $2^{O(f)}$ lightness).
This result was very recently generalized to doubling spaces by Le, Solomon, and Than \cite{LST23}, who obtain $O(f^2)$ lightness (improving over \cite{CLNS15} $O(f^2\log n)$ lightness, and \cite{Solomon14} $O(f^2+f\log n)$ lightness).

Abam \etal\cite{ABFG09} introduced the notion of \emph{region} fault-tolerant spanners for the Euclidean plane. They showed that one can construct a $t$-spanner with $O(n\log n)$ edges in such a way that if points belonging to a convex region are deleted, the residual graph is still  a spanner for the remaining points.

\paragraph{More on Light spanners.}
Light spanners were constructed for high dimensional Euclidean and doubling spaces (in similar context to our \Cref{thm:DdimLarge}) \cite{FN22,LS23}.
Subset light spanners were studied for planar and Minor free graphs \cite{Klein06,Klein08,Le20,CFKL20}, where the goal is to maintain distances only between a subset of terminals (and the lightness is defined w.r.t. the minimum Steiner tree). 
Bartal \etal constructed light prioritized and scaling spanner \cite{BFN19}, where only a small fraction of the vertex pairs suffer from large distortion.
Recently Le and Solomon conducted a systematic study of efficient constructions of light spanners \cite{LS23} (see also \cite{FS20,ADFSW19}).
Finally, light spanners were efficiently constructed in the  LOCAL \cite{KPX08}, and CONGEST \cite{EFN20} distributed models.

\subsection{Organization}
After a few preliminaries in \cref{sec:pre}, we show our reliable spanner for $k$-HSTs in \cref{sec:hst}. In \cref{sec:ppcs-minor-free} we show how to devise PPCS for minor-free graphs, and in \cref{sec:ppcs->hst} we show how to construct light $k$-HST covers based on PPCS. In \cref{sec:results} we combine the results of all previous sections, and derive our results on light reliable spanners for various metric spaces. We show our reliable spanner for the path graph in \cref{sec:path}.
In \cref{sec:minor-free} we devise a reliable spanner for minor-free graphs with improved stretch, and finally, in \cref{sec:lower} we exihibit our lower bounds for the path graph and for ultrametrics.

\section{Preliminaries}\label{sec:pre}

All logarithms (unless explicitly stated otherwise) are in base $2$.
We use $\tilde{O}$ notation to hide poly-logarithmic factors. That is $\tilde{O}(s)=O(s)\cdot\log^{O(1)}(s)$.
For a weighted graph $G = (V, E)$, denote the distance between $u, v \in V$ by $d_G(u, v)$. When $G$ is clear from context, we might write $d(u, v)$.
 For a metric space $(X, d)$, we denote the ball of $v \in X$ of radius $\Delta \ge 0$ by $B(v, \Delta)=\{u\in X~:~d(u,v)\le\Delta\}$.
 The diameter of a cluster $C\subseteq X$ is maximum pairwise distance: $\diam(C)=\max_{u,v\in C}d(u,v)$.

Let $[n]$ denote the set $\{1, \dots, n\}$, and for integers $a\le b$ let $[a:b]$ denote $\{a, \dots, b\}$, and $[a:b)$ denote $\{a,...,b-1\}$.
We next define ultrametrics and HSTs.

\begin{definition}\label{def:HST}
	A metric $(X,d)$ is a called an ultrametric if it satisfies a strong form of the triangle inequality
\[
\forall x,y,z\in X,~d(x,z)\le\max\{d(x,y),d(y,z)\}~.
\] 
Equivalently \cite{BLMN03}, if there exists a bijection $\varphi$ from $X$ to the leaves of a rooted tree $T$ in which:
	\begin{enumerate}
		\item Each node $v\in T$ is associated with a label $\Gamma_{v}$ such that $\Gamma_{v} = 0$ if and only if $v$ is a leaf, and if $u$ is a child of $v$ in $T$ then $\Gamma_{v}\ge\Gamma_{u}$.
		\item $d(x,y) = \Gamma_{\lca(\varphi(x),\varphi(y))}$ where $\lca(u,v)$ is the least common ancestor of $u,v$ in $T$.
	\end{enumerate}

For $k\ge 1$, a $k$-hierarchical well-separated tree ($k$-HST) is an ultrametric $T$ that also satisfies that whenever $u$ is a child of $v$ in $T$, then $\Gamma_{v}\geq k\cdot\Gamma_{u}$.
\end{definition}

\begin{definition}[ultrametric cover]\label{def:hst-cover}
    A \emph{$(\tau,\rho)$-ultrametric cover} for a metric space $(X,d)$, is a collection of at most $\tau$ dominating$^{\ref{foot:dominating}}$ ultrametrics $\mathcal{U} = \{(U_i,d_{U_i})\}_{i=1}^{\tau}$ over $X$, such that for every $x,y\in X$ there is an ultrametric $U_i$ for which $d_{U_i}(x,y)\le \rho\cdot d_X(x,y)$.

    The cover is called {\em $l$-light}, if the weight of every ultrametric $U_i$ is at most $l\cdot w(MST(X))$.
\end{definition}

\begin{definition}[Pairwise Partition Cover Scheme]\label{def:ppcs}
    A collection of partitions $\mathbb{P} = \{\mathcal{P}_{1},\dots,\mathcal{P}_{s}\}$
    is $(\tau,\rho,\eps,\Delta)$-pairwise partition cover if (a) $s\le \tau$, (b) every partition $\mathcal{P}_{i}$ is $\Delta$-bounded (that is, $\forall C\in\cP_i$, $\diam(C)\le\Delta$),
    and (c) for every pair $x,y$ such that $\frac{\Delta}{2\rho}\le d(x,y)\le\frac{\Delta}{\rho}$, there is a cluster $C$ in one of the partitions $\mathcal{P}_{i}$ such that $C$ contains both closed balls $B(x,\eps\Delta),B(y,\eps\Delta)$.\\
    A space $(X,d)$ admits a $(\tau,\rho,\eps)$-\emph{pairwise partition cover scheme} (PPCS) if for every $\Delta>0$, it admits a  $(\tau,\rho,\eps,\Delta)$-pairwise partition cover.
\end{definition}

\begin{definition}[$\Delta$-net]\label{def:net}
    For $\Delta > 0$ and a metric space $(X, d)$, a $\Delta$-net is a set $\mathcal{N} \subseteq X$ such that:
    \begin{enumerate}
        \item Packing: For every $u, v \in \mathcal{N}, d(u, v) > \Delta$
        \item Covering: For every $x \in X$, there exists $u \in \mathcal{N}$ satisfying $d(x, u) \le \Delta$.
    \end{enumerate}
\end{definition}

It is well known that a simple greedy algorithm can find a $\Delta$-net.

\begin{definition}\label{def:doubling}
	A metric space $(X, d)$ has doubling dimension $\ddim$, if for every $r>0$, every ball of radius $2r$ can be covered by $2^{\sddim}$ balls of radius $r$. A family of metrics is called {\em doubling}, if all the metrics in the family have uniformly bounded doubling dimension.
\end{definition}
By applying the definition iteratively, we get the following simple corollary.
\begin{lemma}[Packing Lemma]\label{lem:packing}
	If $(X,d)$ has doubling dimension $\ddim$, and ${\cal N}$ is a $\Delta$-net, then for any $R>1$, a ball of radius $R\cdot \Delta$ contains at most $(2R)^{\sddim}$ net points.
\end{lemma}
The proof uses the fact that a ball of radius $\Delta/2$ cannot contain two net points of ${\cal N}$.

The following lemma is an extension of \cite[Lemma 2]{FL22}, that shows it suffices to bound the expected size and lightness of an oblivious $\nu$-reliable spanner, in order to obtain worst-case guarantees.
\begin{lemma}\label{lem:expect-the-worst}
Suppose that $(X,d)$ admits an oblivious $\nu$-reliable spanner ${\cal D}$ with expected size $m$ and expected lightness $\phi$, then $(X,d)$ admits an oblivious $3\nu$-reliable spanner ${\cal D}'$ with size $3\cdot m$ and lightness $3\cdot \phi$.
\end{lemma}
\begin{proof}
We define ${\cal D}'$ by conditioning on the event $A=\{(|H|\le 3m)\wedge (w(H)\le 3\phi\cdot w(MST(X)))\}$. Observe that $\Pr[|H|>3m]\le 1/3$ and also $\Pr[w(H)> 3\phi\cdot w(MST(X))]\le 1/3$, both by Markov's inequality. So that $\Pr[A]\ge 1/3$.
For any attack $B\subseteq X$, 
\begin{eqnarray*}
\E_{H\sim {\cal D}'}[|B^+_H\setminus B|]&=&\E_{H\sim {\cal D}}[|B^+_H\setminus B|~\mid A]\\
&=&\sum_{H\in\supp({\cal D})}|B^+_H\setminus B|\cdot\frac{\Pr[H\cap A]}{\Pr[A]}\\
&\le&\frac{1}{\Pr[A]}\cdot \sum_{H\in\supp({\cal D})}|B^+_H\setminus B|\cdot\Pr[H]\\
&\le& \frac{\nu\cdot |B|}{\Pr[A]}\le 3\nu\cdot|B|~.
\end{eqnarray*}
\end{proof}

\section{Light Reliable Spanner for $k$-HSTs}\label{sec:hst}

In this section we devise a light reliable spanner for the family of $k$-HSTs (see \Cref{def:HST}).
Let $T$ be the tree corresponding to the given $k$-HST, we refer to its leaves as vertices, and to the interval nodes as nodes. Each node has an arbitrary order on its children.
For a node $x$ we denote by $L(x)$ the set of leaves in the subtree rooted at $x$, and by $L=[n]$ the set of all leaves. For an internal node $x$ in $T$, let $\deg(x)$ denote the number of children of $x$. We will assume that $\deg(x)\ge 2$ (as degree $1$ nodes are never the least common ancestor, and thus can be contracted). Our goal is to prove the following theorem.
\begin{theorem}\label{thm:k-HST-new}
	\sloppy For any parameters $\nu \in (0, 1/6)$ and $k> 1$, every $k$-HST $T$ admits an oblivious $\nu$-reliable $(2 + \frac{2}{k-1})$-spanner of size $n\cdot\tilde{O}(\nu^{-1}\cdot\log\log n)^2$ and lightness $\tilde{O}(\nu^{-1}\cdot\log\log n)^2$.
\end{theorem}

\subsection{Decomposition of $T$ to Heavy Paths} 

We apply the following decomposition of $T$ into paths, reminiscent of the heavy-path decomposition \cite{ST83}. Each node $x\in T$ is given a tag, initially $\sigma_x=|L(x)|$, and set $D=\emptyset$. Go over the nodes of $T$ in preorder, and when visiting node $x$ with children $x_1,\dots,x_t$: If there is $1\le j\le t$ such that $\sigma_{x_j}> (1-\nu/2)\sigma_x$, set $\sigma_{x_j}=\sigma_x$ and add the edge $\{x,x_j\}$ to $D$. 
For example, if $T$ contains a path $(y_1,y_2,\dots,y_q)$ where $y_1$ is the closest vertex to the root, and $L(y_q)>(1-\nu/2)L(y_2)$ while $L(y_2)<(1-\nu/2)L(y_1)$ then it will hold that $\sigma_{y_1}\ne \sigma_{y_2}=\sigma_{y_3}=\dots=\sigma_{y_q}=|L(y_2)|$.

We claim that $\sigma_x\ge |L(x)|$ for every node $x\in T$, because we either have equality or $x$ inherit the original tag of one of its ancestors. As $1-\nu/2>1/2$, there cannot be two different children of $x$ with more than $|L(x)|/2$ leaves in their subtree, hence there can be at most one child $x_j$ for which an edge is added to $D$. So indeed $D$ is a decomposition of $T$ into heavy paths (some paths can be singletons).
Denote by ${\cal Q}$ this collection of paths, and for each $Q\in{\cal Q}$, let $f(Q)$ be the lowest vertex (farthest from the root) on $Q$. We overload this notation, and define $f(x)=f(Q)$, where $Q$ is the heavy path containing $x$. Let $F=\{f(Q)\}_{Q\in{\cal Q}}$ be the set of lowest vertices over all paths.

\begin{claim}\label{claim:heavy-new}
	Each root-to-leaf path $W$ intersects at most $O(\nu^{-1}\log n)$ paths in ${\cal Q}$.
\end{claim}
\begin{proof}
	Fix a path $Q\in {\cal Q}$. Note that all nodes in $Q$ have the same tag $\sigma_Q$. Whenever the path $W$ leaves $Q$, it will go to some node $y$ with $\sigma_y\le (1-\nu/2)\sigma_Q$. The root has tag $n$, so after leaving $2\nu^{-1}\ln n$ heavy paths, the tag will be at most
	\[
	n\cdot(1-\nu/2)^{2\nu^{-1}\ln n}< n\cdot e^{-\ln n} = 1~,
	\]
	since the tag of any internal node $x$ is at least $|L(x)|$, we must have reached a leaf.
\end{proof}

\subsection{Construction}

For each node $y\in F$, we independently sample uniformly at random a set $Z_y$ of $\ell=c\cdot\nu^{-1}\cdot\ln\left(\frac{\ln n}{\nu}\right)$ vertices from $L(y)$, where $c$ is a constant to be determined later. If there are less than $\ell$ vertices in $L(y)$, take $Z_y=L(y)$.
For each internal node $x$ in $T$ with children $x_1, \dots, x_t$, and for every $1\le j\le t$, we add the edges $\{\{y,z\} ~:~ y\in Z_{f(x)},z\in Z_{f(x_j)}\}$ to the spanner $H$.

\paragraph{Defining the set $B^+$.}
Consider an attack $B$. We say that an internal node $x\in T$ is {\em good} if $Z_{f(x)}\setminus B\neq\emptyset$. A leaf $u$ is {\em safe} if for every ancestor $x$ of $u$, $x$ is good. In other words, a leaf is safe if every ancestor $x$ sampled a leaf to $Z_{f(x)}$ which is not in $B$. \\Define $B^+$ as the set of all leaves which are not safe.

\subsection{Analysis}

\paragraph{Size Analysis.}

For each internal node $x$ in $F$ and each child $x_j$ of $x$, we added the bi-clique $Z_x\times Z_{x_j}$, which contains at most $\ell^2$ edges. Since the sum of degrees of internal nodes in $T$ is $O(n)$ (recall that all degrees are at least 2), the total number of edges added to $H$ is $O(n\cdot\ell^2)=n\cdot\tilde{O}(\nu^{-1}\cdot\log\log n)^2$.

\paragraph{Weight Analysis.}
First, we claim that the weight of the MST for the leaves of $T$ is equal to 
\begin{equation}
	\sum_{x\in T}(\deg(x)-1)\cdot\Gamma_{x}~.\label{eq:HSTweight}
\end{equation}
This can be verified by running Boruvka's algorithm, say.\footnote{In Boruvka's algorithm, we start with all vertices as singleton components. In each iteration, every component adds to the MST the edge of smallest weight leaving it (breaking ties consistently). For a $k$-HST, we use a small variation -- only components which are the deepest leaves in the HST participate in the current iteration. We claim that the connected components after the $j$-th iteration correspond to nodes of height $j$ above the leaves. Thus, in the $j$-th iteration, any node $x$ of height $j$ will add $\deg(x)-1$ edges with weight $\Gamma_x$ each, that connect the components corresponding to its children.} 
Every internal node $x$ in $F$, adds at most $\ell^2\cdot\deg(x)$ edges of weight at most $\Gamma_x$ to the spanner. The total weight is thus
\[
\sum_{x\in F}\deg(x)\cdot \ell^{2}\cdot\Gamma_{x}=O(w(MST)\cdot \ell^{2})=w(MST)\cdot\tilde{O}(\nu^{-1}\cdot\log\log n)^2)~.
\]

\paragraph{Stretch Analysis.} 
The stretch analysis is based on the following lemma.
\begin{lemma}\label{lem:stretch-hst}
Let $u\notin B^+$ be any safe leaf. Then for any ancestor $x$ of $u$ and any $v\in Z_{f(x)}\setminus B$, the spanner $H$ contains a path from $u$ to $v$ of length at most $\left(1+\frac{1}{k-1}\right)\cdot\Gamma_x$ that is disjoint from $B$.
\end{lemma}
\begin{proof}
The proof is by induction on $|L(x)|$. The base case is when $x=u$, then $L(u)=\{u\}$ and the statement holds trivially. Let $x$ be an ancestor of $u$, and take any vertex $v\in Z_{f(x)}\setminus B$. We need to find a path in $H$ of length at most $\left(1+\frac{1}{k-1}\right)\cdot\Gamma_x$ from $u$ to $v$ that is disjoint from $B$.

Let $x_u$ be the child of $x$ whose subtree contains $u$. Since $u$ is safe, we know that $Z_{f(x_u)}\setminus B\neq\emptyset$, so take any vertex $u'\in Z_{f(x_u)}\setminus B$.
By the induction hypothesis on $x_j$, there is a path $P'$ in $H$ from $u$ to $u'$ of length at most $\left(1+\frac{1}{k-1}\right)\cdot\Gamma_{x_j}$ disjoint from $B$ (note that indeed $|L(x_j)|<|L(x)|$, as all vertices have degree at least 2). Recall that in the construction step for $x$, we added all edges from $Z_{f(x)}$ to $Z_{f(x_u)}$, in particular the edge $\{u',v\}\in H$. Note that $v\notin B$, that $u',v\in L(x)$ and therefore $d_T(u',v)\le\Gamma_x$, and as $T$ is a $k$-HST we have that $\Gamma_{x_j}\le\frac{\Gamma_x}{k}$. It follows that the path $P=P'\circ\{u',v\}$ from $u$ to $v$ in $H$ is disjoint from $B$, and has length at most
\[
\left(1+\frac{1}{k-1}\right)\cdot\Gamma_{x_j}+\Gamma_x\le \left(\frac{1+\frac{1}{k-1}}{k}\right)\cdot\Gamma_x+\Gamma_x = \left(1+\frac{1}{k-1}\right)\cdot\Gamma_x
\]
\end{proof}

Fix a pair of leaves $u,v\notin B^+$, and let $x=\lca(u,v)$. 
Since both are safe, $Z_{f(x)}\setminus B\neq\emptyset$, and pick any $z\in Z_{f(x)}\setminus B$. By \cref{lem:stretch-hst} there are paths in $H$ from $u$ to $z$ and from $v$ to $z$, both disjoint from $B$, of combined length at most 
\[
2\cdot\left(1+\frac{1}{k-1}\right)\cdot\Gamma_x = \left(2+\frac{2}{k-1}\right)\cdot d_T(u,v)~.
\]

\paragraph{Reliability Analysis.}

For every $x\in T$, denote by $B^{(x)}$ the set of all vertices in $u\in L(x)\setminus B$, such that there is an ancestor $z$ of $u$ in the subtree rooted at $x$ for which $Z_{f(z)}\subseteq B$. In other words, those are the leaves (outside $B$) who are not safe due to a bad ancestor in the subtree rooted at $x$.

We say that a node $x\in T$ is \emph{brutally attacked} if $|B\cap L(x)|\ge(1-\nu)\cdot |L(x)|$, that is at least a $1-\nu$ fraction of the decedent leaves of $x$ are in the attack $B$.
Denote by $B^{(x)}_1\subseteq B^{(x)}$ 
the set of vertices $u\in L(x)\setminus B$ that  have a brutally attacked ancestor $y$ in the subtree rooted at $x$.
Denote by  $B^{(x)}_2=B^{(x)}\setminus B^{(x)}_1$ the rest of the vertices in $B^{(x)}$.

We next argue that the number of vertices added to $B^+$ (in the worst case) 
due to brutally attacked nodes is bounded by $O(\nu)\cdot|B|$.
Let $A_{\rm ba}$ be the set of $T$ nodes which are brutally attacked, and they are maximal w.r.t. the order induced by $T$. That is, $x\in A_\ba$ if and only if $x$ is brutally attacked, while for every ancestor $y$ of $x$, $y$ is not brutally attacked.
Clearly, for every $x\in A_\ba$ it holds that $|B_{1}^{(x)}|\le|L(x)\setminus B|\le\nu\cdot|L(x)|\le\frac{\nu}{1-\nu}\cdot|L(x)\cap B|$.
In total, for the root $r$ of $T$ it holds that 
\[
|B_{1}^{(r)}|=\sum_{x\in A_{\ba}}|B_{1}^{(x)}|\le\sum_{x\in A_{\ba}}\frac{\nu}{1-\nu}\cdot|L(x)\cap B|\le\frac{\nu}{1-\nu}\cdot|B|\le2\nu\cdot|B|~.
\]

Next we bound the damage done (in expectation) due to non brutally attacked nodes. Denote $\beta=\frac{1}{\ln\ln n}$.
We will prove for any node $x\in T$ which is not a heavy child, by induction on $|L(x)|$ that 
\begin{equation}\label{eq:ind-hst}
	\E[|B_2^{(x)}|]\le \max\left\{0,\nu\cdot\beta\cdot\ln\ln(|L(x)|)\cdot |B\cap L(x)|\right\}~.
\end{equation}
The base case  where $|L(x)|\le\nu^{-1}$ holds trivially as $B_2^{(x)}=\emptyset$.
Indeed, consider a descendent leaf $v\notin B$ of $x$.
For every ancestor internal node $y$ of $v$, which is a descendent of $x$, it holds that $f(y)=y$ ($y$ does not have heavy children as $|L(y)|-1=(1-\frac{1}{|L(y)|})\cdot|L(y)|<(1-\frac{\nu}{2})\cdot|L(y)|$). In particular $v\in Z_{f(y)}\setminus B$. It follows that $v\notin B_2^{(x)}$, and thus $B_2^{(x)}=\emptyset$.
In general, 
let $x\in T$ be an inner node, which is not a heavy child.
Denote $m=|L(x)|>\nu^{-1}$. 
$x$ is the first vertex in a heavy path $Q=(x=y_1,y_2,...,y_s)\in{\cal Q}$. Let $x_1,\dots,x_t$ be the children of all the nodes in $Q$. Observe that none of $x_1,\dots,x_t$ is a heavy child, and that $L(x_1),\dots,L(x_t)$ is a partition of $L(x)$.
The main observation is that all the vertices in $Q$ use the same sample $Z_{f(x)}$, so a leaf $u$ is in $B_2^{(x)}$ if at least one the following holds: 
\begin{enumerate}
	\item $u\in B_2^{(x_j)}$ for some $1\le j\le t$, or 
	\item $Z_{f(x)}\subseteq B$.
\end{enumerate} 
We conclude that 

\begin{equation}\label{eq:abcd}
	\E[|B_{2}^{(x)}|]\le\sum_{j=1}^{t}\E[|B_{2}^{(x_{j})}|]+|L(x)|\cdot\Pr[Z_{f(x)}\subseteq B]
	~.
\end{equation}
In what follows we bound each of the two summands. For the first, we use the induction hypothesis on $x_j$ (clearly $|L(x_j)|<m=|L(x)|$), to get that
\[
\mathbb{E}\left[\left|B_{2}^{(x_{j})}\right|\right]\le\max\left\{ 0,\nu\cdot\beta\cdot\ln\ln(|L(x_{j})|)\cdot|B\cap L(x_{j})|\right\} 
~.
\]

By definition of a heavy path, for every $1\le j\le t$, $|L(x_j)|\le (1-\nu/2)\cdot\sigma_Q=(1-\nu/2)\cdot m$. 
It holds that $(1-\frac{\nu}{2})\cdot m\ge(1-\frac{\nu}{2})\cdot\nu^{-1}\ge\nu^{-1}-\frac{1}{2}\ge5.5$, and in particular, $\ln\ln\left((1-\frac{\nu}{2})\cdot m\right)>0$.
It follows that
\begin{eqnarray}\label{eq:abdd}
	\sum_{j=1}^{t}\E[|B_{2}^{(x_{j})}|] & \le\sum_{j=1}^{t}\nu\cdot\beta\cdot\ln\ln\left(\left(1-\frac{\nu}{2}\right)\cdot m\right)\cdot|B\cap L(x_{j})|\nonumber\\
	& =\nu\cdot\beta\cdot\ln\ln\left(\left(1-\frac{\nu}{2}\right)\cdot m\right)\cdot|B\cap L(x)|~.
\end{eqnarray}

For the second summand, we now analyze the probability of the event $Z_{f(x)}\subseteq  B$. 
If $|B\cap L(x)|\ge (1-\nu)\cdot|L(x)|$, then $x$ is brutally attacked and thus $B_2^{(x)}=\emptyset$ and (\ref{eq:ind-hst}) holds. We thus can assume $|B\cap L(x)|< (1-\nu)\cdot|L(x)|$.
By the heavy path decomposition, it holds that $|L(f(x))|>(1-\frac{\nu}{2})\cdot m$. In the case that $|L(f(x))|\le\ell$ we take $Z_{f(x)}=L(f(x))$, and as $|L(f(x))|>(1-\frac{\nu}{2})\cdot m>(1-\nu)m>| B\cap L(x)|$, 
there must be a vertex in $Z_{f(x)}\setminus B$. In particular, $\Pr\left[Z_{f(x)}\subseteq B\right]=0$. Otherwise, we have that  $|L(f(x))|>\ell$.
As $Z_{f(x)}$ is chosen from $L(f(x))$ independently of $ B$,
by \Cref{lem:sterling}, the probability that all of the $\ell$ vertices in $Z_{f(x)}$ are chosen from $ B\cap L(f(x))$ is at most
\begin{align}
	\Pr\left[Z_{f(x)}\subseteq B\right] & =\frac{{|B\cap L(f(x))| \choose \ell}}{{|L(f(x))| \choose \ell}}\le O(\sqrt{\ell})\cdot\left(\frac{|B\cap L(f(x))|}{|L(f(x))|}\right)^{\ell}\nonumber\\
	& \le O(\sqrt{\ell})\cdot\left(\frac{1-\nu}{1-\frac{\nu}{2}}\right)^{\ell-1}\cdot\frac{|B\cap L(f(x))|}{m}\nonumber\\
	& \overset{(*)}{\le}\frac{\nu^{2}\cdot\beta}{4\cdot\ln n}\cdot\frac{|B\cap L(f(x))|}{m}\le\frac{\nu^{2}\cdot\beta}{4\cdot\ln m}\cdot\frac{|B\cap L(x)|}{m}~,\label{eq:HSTdefOfelld}
\end{align}
where the inequality $^{(*)}$ uses that $\frac{1-\nu}{1-\frac{\nu}{2}}\le 1-\frac{\nu}{2}\le e^{-\nu/2}$, and taking a large enough constant $c$ in the definition of $\ell$.
By plugging (\ref{eq:abdd}) and (\ref{eq:HSTdefOfelld}) into (\ref{eq:abcd}) we conclude that,
\begin{align*}
	\mathbb{E}\left[\left|B_{2}^{(x)}\right|\right] & \le\sum_{j=1}^{t}\E[|B_{2}^{(x_{j})}|]+m\cdot\Pr[Z_{f(x)}\subseteq B]\\
	& \le\nu\cdot\beta\cdot\ln\ln\left(\left(1-\frac{\nu}{2}\right)\cdot m\right)\cdot|B\cap L(x)|+\frac{\nu^{2}\cdot\beta}{4\cdot\ln m}\cdot|B\cap L(x)|\\
	& \overset{(**)}{\le}\nu\cdot\beta\cdot\ln\ln m\cdot\left|B\cap L(x)\right|~,
\end{align*}
which concludes the proof of (\ref{eq:ind-hst}), and thus the induction step. It remains to validate $^{(**)}$:

\begin{align*}
	\ln\ln m-\ln\ln\left((1-\frac{\nu}{2})\cdot m\right) & =\ln\frac{\ln m}{\ln\left((1-\frac{\nu}{2})\cdot m\right)}\ge\ln\frac{\ln m}{\ln m-\ln(1+\frac{\nu}{2})}\\
	& \ge\ln\left(1+\frac{\ln(1+\frac{\nu}{2})}{\ln m}\right)\ge\frac{\ln(1+\frac{\nu}{2})}{2\ln m}\ge\frac{\nu}{4\ln m}~,
\end{align*}
using $\ln(1+x)\ge\frac x2$ for $0<x<1$. Finally, by applying (\ref{eq:ind-hst}) on the root $r$ of $T$, we get that 
\[
\E[|B^{+}\setminus B|]=\E[|B_{1}^{(r)}|+|B_{2}^{(r)}|]\le(2\nu+\nu\cdot\beta\cdot\ln\ln n)\cdot|B|=3\nu\cdot|B|~.
\]

 \Cref{thm:k-HST-new} follows by rescaling $\nu$ by a  factor of $3$.

\subsection{Improved Stretch for Small Max Degree HST}
In this subsection we slightly modify \Cref{thm:k-HST-new} to obtain a spanner with stretch $(1 + \frac{2}{k-1})$, while increasing the lightness and sparsity to be linear in the maximum degree of the HST. 
Later, we will use \Cref{thm:k-HST-degree-new} to construct an oblivious light $(1+\eps)$-reliable spanner for doubling metrics.
\begin{theorem}\label{thm:k-HST-degree-new}
	\sloppy Consider a $k$-HST $T$ of maximum degree $\delta$. For any parameters $\nu \in (0, 1/6)$ and $k>1$, $T$ admits an oblivious $\nu$-reliable $(1 + \frac{2}{k-1})$-spanner of size $n\cdot \delta\cdot \tilde{O}\left(\nu^{-1}\cdot\log\log n\right)^{2}$ and lightness $\delta\cdot \tilde{O}(\nu^{-1}\cdot\log\log n)^2$.
\end{theorem}
\begin{proof}
	The construction will follow the exact same lines of \Cref{thm:k-HST-new} with a small tweak. We will use the heavy path decomposition $\cQ$, and for every node $y\in F$, we will sample a set $Z_y$ of size $\ell$ from $L(y)$. The set $B^+$ (and the definition of safe), remain exactly the same.
	The only difference is in the definition of bi-cliques. Specifically, for each internal node $x=x_0$ in $T$ with children $x_1, \dots, x_t$, for every $0\le j< j'\le t$, we add the edges $\{\{y,z\} ~:~ y\in Z_{f(x_j)},z\in Z_{f(x_{j'})}\}$ to the spanner $H$. That is, in addition to adding edges from $Z_{f(x)}$ (the sample set of $x$) to all the other sampled sets (of the children of $x$), we also add all the edges between the two sets $Z_{f(x_j)},Z_{f(x_{j'})}$ of every pair of children of $x$.
	
	As $B^+$ is defined in the exact same way, for every attack $B$ we have $\E[|B^+|]\le(1+\nu)\cdot |B|$.
	
	For the size analysis, consider an internal node $x$ of degree $\deg(x)\le\delta$, we add at most $\ell^2\cdot{\deg(x)+1\choose2}\le \ell^2\cdot\delta\cdot\deg(x)$ edges. In total, the size of the spanner is bounded by $n\cdot \ell^2\cdot\delta = n\cdot\delta\cdot \tilde{O}(\nu^{-1}\cdot\log\log n)^2$.
	
	For the lightness analysis, the total weight added due to an internal node $x$ of degree $\deg(x)\le\delta$ is at most $\ell^2\cdot\delta\cdot\deg(x)\cdot\Gamma_x$. Thus, the total weight added due to the bi-cliques is $\sum_{x\in T}\deg(x)\cdot \ell^{2}\cdot\delta\cdot\Gamma_{x}=\delta\cdot \tilde{O}(\nu^{-1}\cdot\log\log n)^2\cdot w(MST)$.
	
	It remains to analyze the stretch. The argument is similar to \Cref{thm:k-HST-new}, where the main difference is that a $u-v$ path will be using only a single edge in the highest level (instead of two). Note that since we only add additional edges to $H$ in this variant, \Cref{lem:stretch-hst} still holds.
Fix a pair of leaves $u,v\notin B^+$, and let $x=\lca(u,v)$.
Let $x_u$ (resp., $x_v$) be the child of $x$ whose subtree contains $u$ (resp., $v$). Since both $u,v$ are safe, $Z_{f(x_u)}\setminus B\neq\emptyset$ and $Z_{f(x_v)}\setminus B\neq\emptyset$, so pick any $u'\in Z_{f(x_u)}\setminus B$ and $v'\in Z_{f(x_v)}\setminus B$. By the construction step for $x$, we added all edges in $Z_{f(x_u)}\times Z_{f(x_v)}$, in particular, $\{u',v'\}\in H$. Note that $d_T(u',v')\le \Gamma_x$, since both are in $L(x)$.
By \Cref{lem:stretch-hst} there is a path $P_u$ (resp., $P_v$) in $H$ from $u$ to $u'$ (resp., $v$ to $v'$), which is disjoint from $B$, and of length at most $\left(1+\frac{1}{k-1}\right)\cdot\Gamma_{x_u}$ (resp., $\left(1+\frac{1}{k-1}\right)\cdot\Gamma_{x_v}$).  Since $T$ is a $k$-HST we have that $\Gamma_{x_u},\Gamma_{x_v}\le\frac{\Gamma_x}{k}$, therefore the path $P=P_u\circ \{u',v'\}\circ P_v$ is a $u-v$ path in $H$, disjoint from $B$, and has total length at most
\[
2\cdot\left(1+\frac{1}{k-1}\right)\cdot\frac{\Gamma_x}{k} +\Gamma_x = \left(1+\frac{2}{k-1}\right)\cdot d_T(u,v)~.
\]

\end{proof}

\section{Pairwise Partition Cover for Minor Free Graphs}\label{sec:ppcs-minor-free}

In this section we construct a \emph{Pairwise Partition Cover Scheme} (PPCS, recall \Cref{def:ppcs}) for metrics arising from shortest paths of graphs excluding a fixed minor. 
The main building block in the construction of our PPCS is the so called Shortest Path Decomposition (SPD) introduced by \cite{AFGN22}. Roughly speaking, this is a recursive decomposition of the graph into shortest paths, and the measure of interest is the depth of the recursion, as captured by the following definition.

\begin{definition}[$\SPDdepth$]\label{def:SPD-K} A graph has an $\SPDdepth$ $1$ if and only if it is a (weighted) path. A graph $G$ has an $\SPDdepth$ $k \geq 2$ if there exists a \emph{shortest path} $P$, such that deleting $P$ from the graph $G$ results in a graph whose connected components all have $\SPDdepth$ at most $k-1$. 
\end{definition}

It is shown in \cite{AFGN22} that $n$-vertex graphs excluding a fixed minor have $\SPDdepth$ $O(\log n)$ (this follows by using the balanced separator consisting of $O(1)$ shortest paths, by \cite{AG06}).
We will prove the following lemma:
\begin{lemma}\label{lem:pairwise_spddepth}
    For any parameter $0<\varepsilon < 1/6$, any graph $G = (V, E)$ with $\SPDdepth$ $k$ admits a $\left(\frac{k}{\varepsilon}, \frac{2}{1 - 6\varepsilon}, \varepsilon\right)$-PPCS.
\end{lemma}

In particular, as graphs excluding a fixed minor have $\SPDdepth=O(\log n)$, we obtain the following corollary.
\begin{corollary}\label{thm:minor_free:pairwise}
    For any parameter $\varepsilon < 1/6$, every graph $G = (V, E)$ that excludes a fixed minor, admits a $\left(\frac{O(\log n)}{\varepsilon}, \frac{2}{1 - 6\varepsilon}, \varepsilon\right)$-PPCS 
\end{corollary}
\begin{proof}[Proof of \Cref{lem:pairwise_spddepth}]
We will assume for simplicity (and w.l.o.g.) that $\epsilon^{-1}$ is an integer.
    Fix $\Delta > 0$. We will prove  by induction on the $\SPDdepth$, that graphs with $\SPDdepth$ $k$ admit a $\left(\frac{k}{\varepsilon}, \frac{2}{1 - 6\varepsilon}, \varepsilon,\Delta\right)$-PPC, assuming all graphs with $\SPDdepth$ less than $k$ admits a $\left(\frac{k-1}{\varepsilon}, \frac{2}{1 - 6\varepsilon}, \varepsilon,\Delta\right)$-PPC.
    For the base case, we think of a graph with $\SPDdepth$ 0 as the empty graph, where there is nothing to prove.

    Let $G=(V,E)$ be a connected graph with $\SPDdepth$ $k$, denote by $d(u,v)$ the shortest path distance between $u,v\in V$, and let $P$ be a shortest path in $G$ such that every connected component in $G \backslash P$ has $\SPDdepth$ at most $k-1$.
\paragraph{Construction.}
The basic idea is quite simple, we use the $\frac{k-1}{\varepsilon}$ partitions for the connected components of $G\setminus P$, and create $\frac{1}{\varepsilon}$ new partitions, whose goal is proving padding for pairs $u,v$ such that $P$ intersect the shortest $u-v$ path, or the balls $B_G(u, \varepsilon\Delta), B_G(v, \varepsilon\Delta)$.

    We start by defining the new partitions ${\cal P}_{new}=\{\PP_1,\dots,\PP_{\varepsilon^{-1}}\}$. Let $\mathcal{N} = \{z_1, \dots, z_l\} \subseteq P$ be an $\varepsilon\Delta$-net for $P$ (recall \Cref{def:net}). Fix one endpoint of $P$, and assume that $(z_1,\dots,z_l)$ are sorted by their distance to this endpoint of $P$.
    For every $i \in \{0,1,\dots,\varepsilon^{-1}-1\}$, let $\mathcal{N}_i = \{z_j ~:~ j \equiv i \mod \varepsilon^{-1}\}$.
    For every $z_p,z_q \in \mathcal{N}_i$ with $1\le p < q\le l$,  we have that
    \[d(z_p, z_q) = \sum_{j=p}^{q-1}{d(z_j, z_{j+1})} > (p - q)\varepsilon\Delta \ge\Delta ~.\] 
    The equality holds as $P$ is a shortest path, the first inequality holds since the distance between net points is larger than $\varepsilon\Delta$, and the last inequality by definition of ${\cal N}_i$. Thus, the balls $B(z_p, \Delta/2), B(z_q, \Delta / 2)$ are disjoint.

    For every $0\le i \le \varepsilon^{-1}-1$, we set $\PP_i$ to contain the clusters $\{B(z, \Delta/2)\}_{z \in \mathcal{N}_i}$, and add the rest of the vertices (those that are not contained in any of these balls) as singleton clusters.

    Let $G_1, \dots, G_t$ be the connected components of $G \backslash P$, where $t$ is the number of connected components.
    For every $1\le j\le t$, we apply the induction hypothesis  on $G_j$, which yields a   $\left(\frac{k-1}{\varepsilon}, \frac{2}{1 - 6\varepsilon}, \varepsilon,\Delta\right)$-PPC for $G_j$. This is a collection $\cF^{(j)} = \{\PP_1^{(j)}, \dots \PP_{\varepsilon^{-1}(k-1)}^{(j)}\}$ of $\varepsilon^{-1}(k-1)$ partitions. 
    For every $1\le i \le \varepsilon^{-1}(k-1)$, we construct a partition $\HH_i$ for $G$, by taking $\cup_{j=1}^t \PP^{(j)}_i$, and adding the remaining vertices (note these are the vertices of $P$) as singleton clusters. 
    We return $\cF=\{\PP_{i}\}_{i=0}^{\eps^{-1}-1} \cup \{\HH_i\}_{1\le i \le \varepsilon^{-1}(k-1)}$ as the PPC for $G$. It remains to show that $\cF$ is indeed a $\left(\frac{k}{\varepsilon}, \frac{2}{1 - 6\varepsilon}, \varepsilon, \Delta\right)$-PPC.

    \paragraph{Correctness.}
    First observe that $\cF$ is a set of partitions: for $0\le i\le\varepsilon^{-1}-1$, $\PP_i$ is a partition by definition, while for $1\le i\le \varepsilon^{-1}\cdot(k-1)$, $\HH_i$ is a partition since the connected components are pairwise disjoint. The number of partitions is $\varepsilon^{-1} + \varepsilon^{-1}(k-1) = \varepsilon^{-1} \cdot k$ as required.

    \paragraph{Diameter bound.} Note that $\PP$ is $\Delta$-bounded, because every cluster is either a ball of radius $\Delta/2$, a singleton, or a cluster in a $\Delta$-bounded partition $\HH_i$.
    
\paragraph{Padding property.}
    Let $u, v \in V$, and denote by $P_{uv}$ the shortest $u-v$ path in $G$, and by $B_u=B(u, \varepsilon\Delta)$, $B_v=B(v, \varepsilon\Delta)$. If $\Delta>0$ is such that $\frac{(1 - 6\varepsilon)\Delta}{4}\le d(u, v) \le \frac{(1 - 6\varepsilon)\Delta}{2}$, then we need to show that at least one of the partitions in $\PP$ contains a cluster $C$ such that both $B_u, B_v$ are contained in $C$.

Suppose first that $P$ is disjoint from $P_{uv}\cup B_u\cup B_v$.
   In this case, there exists a connected component $G_j$ in $G \backslash P$, such that $B_u\cup B_v\cup P_{uv} \subseteq G_j$, and therefore $d_{G_j}(u, v) = d(u, v)$. Thus, by the induction hypothesis, there is a cluster $C$ in $\cF^{(j)}$ which contains both $B_u,B_v$, and this cluster is also in one of the $\HH_i$, and thus in $\cF$. (While in general, distances in $G_j$ can be larger from those of $G$, the balls $B_u,B_v$ and $P_{uv}$ remain exactly the same, as they are disjoint from $P$.)

    Consider now the case (see \Cref{fig:minor_free:pairwise} (a)), where $P$ intersects $P_{uv}$. Let $x\in P\cap P_{uv}$  be an (arbitrary) vertex in the intersection. By the covering property of nets, there exists $z \in \mathcal{N}$ such that $d(x, z) \le \varepsilon \Delta$. We bound the distance from any $y \in B_u$ to $z$ by the triangle inequality,
    \begin{align*}
        d(z, y) & \le d(z, x) + d(x, u) + d(u, y)                                            \\
                  & \le d(z, x) + d(v, u) + d(u, y)                                            \\
                  & \le \varepsilon\Delta + \frac{(1 - 6\varepsilon)\Delta}{2} + \varepsilon\Delta
        \le \Delta / 2.
    \end{align*}

   Thus, the cluster $C=B(z,\Delta/2)$ satisfies $B_u \subseteq C$ and by a symmetric argument $B_v \subseteq C$, as required.

    The remaining case is that $P$ intersects $B_u$ or $B_v$. Assume w.l.o.g.\ $P$ intersects $B_v$, and let $x\in P\cap B_v$ (see \Cref{fig:minor_free:pairwise} (b)). As before, there exists $z \in \mathcal{N}$ such that $d(x, z) \le \varepsilon \Delta$.

    Let $y \in B_u$. By the triangle inequality

    \begin{align*}
        d_G(z, y) &\le d_G(z, x) + d_G(x, v) + d_G(v, u) + d_G(u, y) \\
        &\le \varepsilon\Delta + \varepsilon\Delta + \frac{(1 - 6\varepsilon)\Delta}{2} + \varepsilon\Delta \le \Delta/2,
    \end{align*}
    hence $B_u\subseteq C:=B(z, \Delta/2)$. The argument for $B_u$ is simpler. So both balls are in the same cluster $C$, as required.
\end{proof}
\begin{figure}[htbp]
    \centering
    \includegraphics[width=0.7\textwidth]{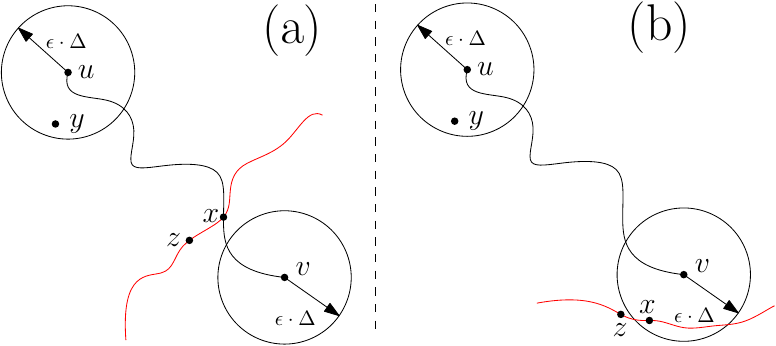}
    \caption{\it Illustration of the proof of \Cref{lem:pairwise_spddepth}, where we show if $P$ (colored red) intersects either $P_{uv}$ (figure (a)) or any of $B_u,B_v$ (figure (b)), then both $B_u, B_v$ are in $B(z, \Delta / 2)$.
    \label{fig:minor_free:pairwise}}
\end{figure}

\section{From Pairwise Partition Cover to Light $k$-HST Cover}\label{sec:ppcs->hst}

In this section we devise a light $k$-HST cover (see \Cref{def:hst-cover}) from a Pairwise Partition Cover Scheme (PPCS, see \Cref{def:ppcs}). The framework essentially follows that of \cite{FL22}, except that we need to guarantee also a bound on the lightness of each tree in the cover. To this end, we ensure each cluster contains a net point (recall \Cref{def:net}).

The following simple claim, which lower bounds the MST weight with respect to a net, is proven in \cite[Claim 1]{FN22}.
\begin{claim}[\cite{FN22}]\label{claim:net:light}
    Let $\mathcal{N}$ be a $\Delta$-net of a metric space $(X, d)$. Then $|\mathcal{N}| \le\left\lceil \frac{2}{\Delta}\cdot w(\MST(X))\right\rceil $.
\end{claim}
The main result of this section is captured by the following theorem.
\begin{theorem}\label{thm:pairwise_partition_cover_to_ultrametric_cover}
Fix any integer $\tau\ge 1$, and parameters $\rho\ge 1$ and $0<\epsilon<1/12$.
    Suppose that a given metric space $(X,d)$ admits a $(\tau,\rho,\epsilon)$-PPCS, then for any $k \ge \frac{8\rho}{\varepsilon}$, $(X,d)$  admits a $O(k\log n)$-light  $\left(O(\frac{\tau}{\epsilon}\log{k}),\rho(1+3\epsilon)\right)$-$k$-HST cover.
\end{theorem}

    Assume w.l.o.g. that the minimal distance in $(X,d)$ is $1$, and let $\Phi$ be the maximal distance. Fix a real number  $1\le l\le k$, and for $-1\le i\le \log_k\Phi$, let $\Delta_i(l) = l\cdot k^i$ (for brevity we will omit $l$ 
    when it is clear from context), and let $\mathcal{N}_i$ be an $\frac{\varepsilon\Delta_i}{4}$-net. The following lemma shows how to change a collection of pairwise partition covers, so it will become hierarchical and each cluster will contains a net point.

 \begin{lemma}\label{lem:ppcs-hier-light}
    Fix a real number  $1\le l\le k$. For each integer $-1\le i\le \log_k\Phi$, let $\{\PP_1^i, \dots \PP_\tau^i\}$ be a $(\tau,\rho,\epsilon, \Delta_i)$-pairwise partition cover.
    Then there exists a collection of $(\tau, (1 + \varepsilon)\rho, 0, (1 + \varepsilon)\Delta_i)$-pairwise partition covers $\{\tilde{\PP}_1^i, \dots \tilde{\PP}_\tau^i\}_{i=-1}^{\log_k\Phi}$ that satisfies the following two properties:
    \begin{enumerate}
    \item For every $-1\le i\le \log_k\Phi$ and $1\le j\le\tau$, $|\tilde{\PP}_j^i| \le |\mathcal{N}_i|$.
    
        \item For every  $1\le j\le\tau$, the partitions ${\{\tilde{\PP}_j^i \}}_{i \ge -1}$ are hierarchical (that is, for each $0\le i\le \log_k\Phi$, every cluster of $\tilde{\PP}_j^{i-1}$ is contained in a cluster of $\tilde{\PP}_j^i$).   
    \end{enumerate}
    \end{lemma}
    \begin{proof}
    Fix $j \in [\tau]$. We show how to construct ${\{\tilde{\PP}_j^i \}}_{i \ge -1}$ by induction on $i$.
    For $i=-1$, since $\Delta_{-1}= l/k\le 1$, there is no padding requirement, and we may take the trivial partition to singletons. Assume that for some $0\le i\le\log_k\Phi$, we constructed $\tilde{\PP}_j^{i-1}$ that satisfies both properties, and we will show how to construct $\tilde{\PP}_j^{i}$.

    Start with the partition ${\PP}_j^{i}$. The first change will force every cluster to contain a net point.
    For each cluster $C \in\PP_j^{i}$, if $C \cap \mathcal{N}_i = \emptyset$, we remove $C$ from ${\PP}_j^{i}$. Then for every $v \in C$ we add $v$ to the cluster in ${\PP}_j^{i}$ containing the nearest net point in $\mathcal{N}_i$ to $v$. This creates a partition ${\hat P}_j^{i}$. 
    Now every cluster contains at least one net point, therefore $|{\hat P}_j^i| \le |\mathcal{N}_i|$. Also observe that the new cluster of $v$ will not be removed.

    The second change will guarantee the hierarchical property. For each cluster $C' \in \tilde{\PP}_j^{i-1}$, move all the vertices of $C'$ to some cluster $C \in {\hat P}_j^{i}$ which intersects $C'$. Call the resulting partition $\tilde{\PP}_j^{i}$, which satisfies the second property by construction.

Observe that it is no longer true that every cluster of $\tilde{\PP}_j^{i}$ contains a net point (it could have moved in the second change). Nevertheless, the number of clusters in $\tilde{\PP}_j^{i}$ did not change.
It remains to show that $\{\tilde{\PP}_1^i, \dots \tilde{\PP}_\tau^i\}$ is indeed a $(\tau, (1 + \varepsilon)\rho, 0, (1 + \varepsilon)\Delta_i)$-pairwise partition cover.

    \paragraph{Diameter bound.}
    We start by showing that each cluster $\tilde{C} \in \tilde{\PP}_j^{i}$ has diameter at most $(1 + \varepsilon)\Delta_i$, by induction on $i$. The base case $i=-1$ is trivial since every cluster has diameter $0$.
    Assume the claim holds for $i-1$ and we will prove it for $i$.

    Let $C \in {\PP}_j^{i}$ be the cluster before the updates leading to $\tilde{C}$. In the first change we may have moved vertices from other clusters (those without a net point) to $C$, creating the cluster $\hat{C}$. By the covering property of nets, these vertices are at distance most $\frac{\varepsilon\Delta_i}{4}$ from some net point in $C$.
      For any $u\in \hat{C}$, let $r_u\in C$ be the closest point to $u$ in $C$ (not necessarily a net point). Then for any $u,v\in\hat{C}$,
      \begin{equation}\label{eq:ddiam}
      d(u,v)\le d(u,r_{u})+d(r_{u},r_{v})+d(r_{v},v)\le\frac{\varepsilon\Delta_{i}}{4}+\diam(C)+\frac{\varepsilon\Delta_{i}}{4}=\diam(C)+\frac{\varepsilon\Delta_{i}}{2}~.
      \end{equation}
    In particular,  $\diam(\hat{C})\le\diam(C)+\frac{\varepsilon\Delta_i}{2}$.

    In the second change, we may have added to $\hat{C}$ entire clusters $C' \in \tilde{\PP}_j^{i-1}$ which intersect it, creating $\tilde{C}$ (note that we may have also removed points from $C$, but this surely will not increase the diameter). The diameter of each $C'$ is at most $(1 + \varepsilon)\Delta_{i-1}$ by the induction hypothesis. 
    Hence, by a similar argument to above, \[\diam(\tilde{C})\le\diam(\hat{C})+2\diam(C')\le\diam(\hat{C})+2(1 + \varepsilon)\Delta_{i-1}~.\]  
    Recall that $k\ge 8\rho/\varepsilon\ge (1 + \varepsilon)4/\varepsilon$, and so $2(1 + \varepsilon)\Delta_{i-1}=2(1 + \varepsilon)\Delta_i/k\le \varepsilon\Delta_i/2$.   
    We conclude that 
    \[
        \diam(\tilde{C})\le \diam(\hat{C})+\frac{\varepsilon\Delta_i}{2} \le \diam(C)+2\cdot\frac{\varepsilon\Delta_i}{2}\le (1+\varepsilon)\cdot\Delta_i~.
    \]

\paragraph{Padding property.}
    It remains to show that for $u, v \in X$, if there exists $-1\le i\le\log_k\Phi$ such that $\frac{\Delta_i}{2\rho} = \frac{(1 + \varepsilon)\Delta_i}{2(1 + \varepsilon)\rho} \le d(u, v) \le \frac{(1 + \varepsilon)\Delta_i}{2(1 + \varepsilon)\rho} = \frac{\Delta_i}{\rho}$, then both $u,v$ are contained in a single cluster in at least one of the partitions $\{\tilde{\PP}_1^{i},...,\tilde{\PP}_\tau^{i}\}$. By the padding property of $\{{\PP}_1^{i},...,{\PP}_\tau^{i}\}$, there exists $1\le j\le\tau$ and a cluster $C\in \PP_j^i$, such that $B(u, \varepsilon \Delta_i), B(v, \varepsilon \Delta_i) \subseteq C $.
    We argue that $u, v \in \tilde{C}$ for the cluster $\tilde{C}\in \tilde{\PP}_j^i$ created from $C$ by our construction.
    
    By the covering property of nets, there is a net point of ${\cal N}_i$ in $B(u, \varepsilon \Delta_i) \subseteq C$, thus $C$ was not removed in the first change, and there is a corresponding cluster $\hat{C}\in\hat{\PP}_j^{i}$ (note that $C\subseteq \hat{C})$. 
    
    Let $\tilde{C}_u, \tilde{C}_v \in \tilde{\PP}_j^{i-1}$ be the clusters containing $u, v$ respectively. The diameter of $\tilde{C}_u, \tilde{C}_v$ is bounded by $(1 + \varepsilon)\Delta_{i-1} = (1 + \varepsilon)\cdot\frac{\Delta_i} {k}\le \frac{(1 + \varepsilon)\varepsilon}{8\rho}\Delta_i < \varepsilon \Delta_i$. Thus, these clusters are contained in $B(u, \varepsilon\Delta_i), B(v, \varepsilon\Delta_i)$ respectively, and therefore also in $\hat{C}$.
    So after the second change, $u,v$ do not move to any other cluster, and are both in $\tilde{C}$. 
    
    This concludes the proof that $\{\tilde{\PP}_1^i, \dots \tilde{\PP}_\tau^i\}$ is a $(\tau, (1 + \varepsilon)\rho, 0, (1 + \varepsilon)\Delta_i)$-pairwise partition cover.

\end{proof}
We are now ready to prove the main theorem of this section.
\begin{proof}[Proof of \Cref{thm:pairwise_partition_cover_to_ultrametric_cover}]
Fix $l \in \{{(1 + \varepsilon)}^c ~:~ c \in [0, \log_{1 + \varepsilon}{k}]\}$. Since $(X,d)$ admits a PPCS, for every integer $i\ge-1$ there exist $\{\PP_1^i, \dots \PP_\tau^i\}$ that is a $(\tau,\rho,\epsilon, \Delta_i)$-pairwise partition cover. Apply \Cref{lem:ppcs-hier-light} to obtain a $(\tau, (1 + \varepsilon)\rho, 0, (1 + \varepsilon)\Delta_i)$-pairwise partition cover $\{\tilde{\PP}_1^i, \dots \tilde{\PP}_\tau^i\}$ that satisfy both properties described in the lemma.  

For every $j \in [\tau]$ we construct a single $k$-HST $T$ from the collection of partitions ${\{\tilde{\PP}_j^i \}}_{-1\le i\le\log_k\Phi}$. There is a bijection from the nodes of $T$ to the clusters of the partitions. The leaves of $T$ correspond to the singleton clusters of $\tilde{\PP}_j^{-1}$. For each $0\le i \le\log_k\Phi$, and each cluster $C\in \tilde{\PP}_j^i$, create a node $x=x(C)$  with label $\Gamma_x=(1+\varepsilon)\cdot\Delta_i$, and connect $x$ to all the nodes corresponding to the clusters $\{C' \subseteq C~:~ C'\in \tilde{\PP}_j^{i-1}\}$ (here we use the fact that this pairwise partition cover is hierarchical). Since the label of every such $C'$ is $(1+\varepsilon)\cdot\Delta_{i-1}=(1+\varepsilon)\cdot\Delta_i/k$, and the distance between every two points in $C$ is at most $(1+\varepsilon)\cdot\Delta_i$, this $T$ is indeed a dominating $k$-HST.

We construct $\tau$ of these $k$-HSTs for every $l$, and the collection of all these is our $k$-HST cover for $(X,d)$. The number of $k$-HSTs is indeed $\tau\cdot(1+\log_{1 + \varepsilon}{k}) = O(\frac\tau\eps\cdot \log k)$, as required. It remains to bound the lightness of each $T$, and argue about the stretch of this cover.
     
\paragraph{Lightness bound.}
Now we show that for any $k$-HST $T$ created as above, its lightness is $O(k\log n)$. Recall that the weight of $T$ is $\sum_{x\in T}(\deg(x)-1)\cdot\Gamma_x$ (see equation (\ref{eq:HSTweight})). For any $0\le i\le \log_k\Phi$, by construction the sum of degrees of nodes corresponding to clusters of $\tilde{\PP}_j^i$ is exactly equal to $|\tilde{\PP}_j^{i-1}|$. By the first property of the lemma we have that $|\tilde{\PP}_j^{i-1}|\le |{\cal N}_{i-1}|$, so
\begin{eqnarray*}
w(T)&=&\sum_{x\in T}(\deg(x)-1)\cdot\Gamma_x\\
&\le& \sum_{i=0}^{\log_k\Phi}|\tilde{\PP}_j^{i-1}|\cdot\Delta_i\\
&\le&k\cdot \sum_{i=0}^{\log_k\Phi}|{\cal N}_{i-1}|\cdot\Delta_{i-1}
\end{eqnarray*}

Denote $W=w(MST(X))$. If $\Phi\ge n^3$, we bound separately the lower terms in the sum,
\begin{eqnarray*}
 k\cdot\sum_{i=0}^{\log_k(\Phi/n^3)}|{\cal N}_{i-1}|\cdot\Delta_{i-1}&\le& k\cdot \sum_{i=0}^{\log_k(\Phi/n^3)}n\cdot l\cdot k^i\\
 &\le& 2n\cdot k^2\cdot (\Phi/n^3)\\
 &\le& 2W~,
\end{eqnarray*}
using that $l\le k$ and $\Phi\le W$. For the remaining terms, we have by \Cref{claim:net:light} that $|{\cal N}_i|\cdot\Delta_i=O(W)$, therefore
\begin{eqnarray*}
k\cdot\sum_{i=\max\{0,\log_k(\Phi/n^3)\}}^{\log_k\Phi}|{\cal N}_{i-1}|\cdot\Delta_{i-1}&\le& k\cdot  \sum_{i=\max\{0,\log_k(\Phi/n^3)\}}^{\log_k\Phi} O(W)\\
&=& O(k\cdot\log n\cdot W)~,
\end{eqnarray*}
so the lightness of each tree is indeed $O(k\log n)$.

    \paragraph{Stretch bound.}
    Fix any $u, v \in X$, and let $D=\rho\cdot(1+\varepsilon)\cdot d(u,v)$. Let $i=\lfloor\log_kD\rfloor$, and note that $k^i\le D<k^{i+1}$, so there exists integer  $0\le c\le \log_{1+\varepsilon}k$ such that $l\cdot k^i\le D<(1+\varepsilon)\cdot l\cdot k^i$ (recall that $l=(1+\varepsilon)^c$).
    With these choices of $l$ and $i$ we get that
    \[
       \frac{\Delta_i}{2\rho} \le \frac{\Delta_i}{\rho\cdot(1 + \varepsilon)} \le  d(u, v) \le \frac{\Delta_i}{\rho}.
    \]
    By the padding property of $\{\tilde{\PP}_j^i\}_{1\le j\le\tau}$, there exists $j \in [\tau]$ and a cluster $C \in \tilde{\PP}_j^i$ such that $u, v \in C$. So in the $k$-HST $T$ created from $\tilde{\PP}_j^i$, there is a node $x$ corresponding to $C$ with $\Gamma_x=(1+\varepsilon)\Delta_i$, and so \[d_T(u, v) \le  (1 + \varepsilon)\Delta_i \le \rho\cdot{(1+\varepsilon)^2} \cdot d(u, v) \le \rho\cdot{(1+3\varepsilon)} \cdot d(u, v)~.\]

\end{proof}

\subsection{$k$-HST Cover for Doubling Metrics.}

The following lemma asserts that in our construction of $k$-HST cover described above, every tree has bounded degree.
\begin{lemma}\label{lem:pairwise-doubling-bounded}
	If a metric space $(X,d)$ has doubling dimension $\ddim$, then every $T$ in the $k$-HST cover of \Cref{thm:pairwise_partition_cover_to_ultrametric_cover} has maximum degree $O(k/\varepsilon)^{\sddim}$.
\end{lemma}
\begin{proof}
	Let $x\in T$ be any node with children $x_1,...,x_t$. The node $x$ corresponds to a cluster $\tilde{C} \in \tilde{\PP}^i_j$, and its children to clusters $\tilde{C}_1, \dots, \tilde{C}_t \in \tilde{\PP}^{i-1}_j$ contained in $\tilde{C}$.
	Recall that in the partition $\hat{\PP}^{i-1}_j$, 
	every cluster contains a net point from an $\varepsilon\Delta_{i-1}/4$-net ${\cal N}_{i-1}$. Since every cluster of 
	$\tilde{\PP}^{i-1}_j$ was a cluster of $\hat{\PP}^{i-1}_j$, the clusters $\tilde{C}_1, \dots, \tilde{C}_t$ correspond to different net points. The maximal distance between any two such net points is 
	\[
	\diam(\tilde{C})+2\varepsilon\Delta_{i-1}/4<2\Delta_i~,
	\]
	so all these net points are contained in a ball of radius $2\Delta_i$. Since $\Delta_{i-1}=\Delta_i/k$, by the packing lemma (\Cref{lem:packing}) we conclude that $t\le O(k/\varepsilon)^{\ddim}$.

\end{proof}

 Filtser and Le \cite{FL22} constructed a PPCS for doubling metrics:
\begin{lemma}[\cite{FL22}]\label{lem:pairwise-partition-doubling}
	Every metric space $(X,d)$ with doubling dimension $\ddim$ admits an  $(\eps^{-O(\sddim)}, 1+\eps,\eps)$-pairwise partition cover scheme for any $\eps \in (0,1/16)$.
\end{lemma}

By applying \Cref{thm:pairwise_partition_cover_to_ultrametric_cover} (and using \Cref{lem:pairwise-doubling-bounded}), we conclude
\begin{corollary}\label{cor:HSTcoverDoubling}
	\sloppy For any $\eps \in (0,1/16)$, every $n$-point metric space $(X,d)$ with doubling dimension $\ddim$ admits an $O(\eps^{-1}\cdot\log n)$-light $(\eps^{-O(\sddim)},1+\eps)$-$\frac{16}{\eps}$-HST cover, furthermore, the maximum degree of any tree in the cover is $\eps^{-O(\sddim)}$.
\end{corollary}
\begin{proof}
	Using \Cref{lem:pairwise-partition-doubling}, consider a $(\eps^{-O(\sddim)}, 1+\eps,\eps)$-PPCS for $X$.
	Fix $k=\frac{16}{\eps}$.
	By \Cref{thm:pairwise_partition_cover_to_ultrametric_cover}, $X$ admits a $O(\eps^{-1}\cdot\log n)$-light $(\eps^{-O(\sddim)},1+O(\eps))$-$k$-HST cover.
	Furthermore, by \Cref{lem:pairwise-doubling-bounded}, every HST in the cover has maximum degree $O(\frac{k}{\eps})^{\sddim}=\eps^{-O(\sddim)}$.
	The corollary follows by rescaling $\eps$ accordingly.
\end{proof}

\section{Reliable Spanners for Metric Spaces}\label{sec:results}
We begin this section by proving a meta theorem, which given a light $k$-HST cover, constructs an oblivious light reliable spanner. In the following subsections, we will apply this meta-theorem to obtain the main results of the paper.

\begin{theorem}[Light Reliable Spanner from Light HST Cover]\label{thm:ultrametric_cover_to_reliable_spanner}
	Consider an $n$ point metric space $(X,d)$ that admits $\psi$-light $(\tau, \rho)$-$k$-HST cover $\cT$, for some $k>1$.
	Then for every parameter $\nu \in (0, 1/6)$, $X$ admits an oblivious $\nu$-reliable $(2 +\frac{2}{k-1})\cdot \rho$-spanner of size $n\cdot \tilde{O}\left(\tau^{3}\cdot(\nu^{-1}\cdot\log\log n)^2)\right)$ and lightness $\psi\cdot\tilde{O}(\tau^{3}\cdot(\nu^{-1}\cdot\log\log n)^2)$.
\end{theorem}
\begin{proof}
	For every $k$-HST $T\in\cT$, using \Cref{thm:k-HST-new}
	we construct a $\nu'$-reliable spanner $H_T$ for $T$ for $\nu'=\frac{\nu}{\tau}$. The final spanner we return is $H=\cup_{T\in\cT}H_T$.
	By \Cref{thm:k-HST-new}, the size of the spanner is $|H|=\tau\cdot n\cdot \tilde{O}(\nu'^{-1}\cdot\log\log n)^2=n\cdot \tilde{O}\left(\tau^{3}\cdot(\nu^{-1}\cdot\log\log n)^2\right)$, while the lightness is 
	\begin{align*}
		w(H)\le\sum_{T\in\cT}w(H_{T}) & \le\sum_{T\in\cT}\tilde{O}(\nu'^{-1}\cdot\log\log n)^2\cdot w(MST(T))\\
		& \le\psi\cdot\tilde{O}(\tau^{3}\cdot(\nu^{-1}\cdot\log\log n)^2)\cdot w(MST(X))
	\end{align*}
	Consider an attack $B\subseteq X$. For every spanner $H_T$, let $B^+_T$ be the respective super set, and denote $B^+=\cup_{T\in\cT}B^+_T$. It holds that 
	\[
	\E\left[\left|B^{+}\setminus B\right|\right]\le\sum_{T\in\cT}\E\left[\left|B_{T}^{+}\setminus B\right|\right]\le\tau\cdot\nu'\cdot|B|=\nu\cdot|B|~.
	\]
	Finally, consider a pair of points $u,v\notin B^+$. The is some $k$-HST $T\in\cT$ such that $d_T(u,v)\le\rho\cdot d_X(u,v)$. As $u,v\notin B^+_T$, it holds that 
	\[
	d_{H\setminus B}(u,v)\le d_{H_{T}\setminus B}(u,v)\le(2+\frac{2}{k-1})\cdot d_{T}(u,v)\le(2+\frac{2}{k-1})\cdot\rho\cdot d_{X}(u,v)~.
	\]
\end{proof}
By using \Cref{thm:k-HST-degree-new} instead of \Cref{thm:k-HST-new} in the proof of \Cref{thm:ultrametric_cover_to_reliable_spanner} (and keeping all the rest intact) we obtain:
\begin{corollary}\label{cor:ultrametric_cover_to_reliable_spanner_doubling}
	Consider an $n$ point metric space $(X,d)$ that admits $\psi$-light $(\tau, \rho)$-$k$-HST cover $\cT$, where all the trees in $\cT$ have maximum degree $\delta$.
	Then for every parameter $\nu \in (0, 1/6)$, $X$ admits an oblivious $\nu$-reliable $(1 +\frac{2}{k-1})\cdot \rho$-spanner of size $n\cdot\delta\cdot\tilde{O} \left(\tau^{3}\cdot(\nu^{-1}\cdot\log\log n)^2\right)$ and lightness $\psi\cdot\delta\cdot\tilde{O}(\tau^{3}\cdot(\nu^{-1}\cdot\log\log n)^2)$.
\end{corollary}

\subsection{Doubling Metrics}
By applying \Cref{cor:ultrametric_cover_to_reliable_spanner_doubling}, on the HST cover of \Cref{cor:HSTcoverDoubling} (and rescaling $\eps$) we obtain:

\begin{corollary}\label{thm:doubling}
	For any $\eps,\nu \in (0,1/16)$, every $n$-point metric space $(X,d_X)$ with doubling dimension $\ddim$ admits
	$\nu$-reliable $(1 + \eps)$-spanner with size $n\cdot\eps^{-O(\sddim)}\cdot\tilde{O}(\nu^{-1}\cdot\log\log n)^2$, and lightness $\eps^{-O(\sddim)}\cdot\tilde{O}(\nu^{-2}\cdot\log n)$.	
\end{corollary}
Note that the shortest path metric of the path graph has doubling dimension $1$. 
Hence the lower bound of \Cref{thm:PathObliviousLB} apply. In particular, for constant $\ddim$ and $\eps$,  \Cref{thm:doubling} is tight up to lower order terms.

\subsection{General Metric Spaces}
In this subsection we construct oblivious light reliable spanner for general metric spaces. We begin with the pairwise partition cover of Filtser and Le \cite{FL22}.
\begin{lemma}[\cite{FL22}]\label{lem:pairwise-partition-general}
	Every $n$-point metric space $(X,d_X)$ admits an 
	$(O(n^{1/t}\log n), 2t+\eps,\frac{\varepsilon}{2t(2t+\eps)})$-PPCS for any $\eps \in [0,1]$ and integer $t\ge 1$.
\end{lemma}
By applying \Cref{thm:pairwise_partition_cover_to_ultrametric_cover}, we conclude
\begin{corollary}\label{cor:HSTcoverGeneralMetric}
	\sloppy Every $n$-point metric space $(X,d_X)$ admits a $O(\eps^{-1}\cdot t^{3}\cdot\log n)$-light $\left(n^{1/t}\cdot\log n\cdot\tilde{O}(\frac{t^{2}}{\eps}),2t+\eps\right)$-$\frac{200\cdot t^3}{\eps}$-HST cover 	
	for any $\eps \in (0,1/3)$ and integer $t\ge 1$.
\end{corollary}
\begin{proof}
	Using \Cref{lem:pairwise-partition-general}, consider a $(O(n^{1/t}\log n), 2t+\eps,\frac{\eps}{2t(2t+\eps)})$-PPCS for $X$.
	Fix $k=\frac{8\cdot(2t+\eps)}{\frac{\eps}{2t(2t+\eps)}}=\frac{16t\cdot(2t+\eps)^{2}}{\eps}\ge\frac{64t^{3}}{\eps}$.
	Note that $k=O(\eps^{-1}\cdot t^{3})$.
	By \Cref{thm:pairwise_partition_cover_to_ultrametric_cover} (note that indeed $\frac{\eps}{2t(2t+\eps)}<1/12$), $X$ admits a $\phi$-light $(\tau,\rho)$-$k$-HST cover for
	\begin{align*}
		\phi= & O(k\log n)=O(\eps^{-1}\cdot t^{3}\cdot\log n)\\
		\tau= & O(\frac{n^{1/t}\log n}{\frac{\eps}{2t(2t+\eps)}}\cdot\log k)=O(n^{1/t}\cdot \eps^{-1}\cdot t^{2}\cdot\log n\cdot\log(t/\eps))=n^{1/t}\cdot\log n\cdot\tilde{O}(t^{2}/\eps)\\
		\rho= & (2t+\eps)(1+\frac{3\eps}{2t(2t+\eps)})=2t+\eps+\frac{3\eps}{2t}<2t+3\eps~.
	\end{align*}
	The corollary follows by rescaling $\eps$ by 3, and noting that every $k$-HST is also a $\frac{64t^3}{\eps}$-HST. 
\end{proof}

By applying \Cref{thm:ultrametric_cover_to_reliable_spanner} on the HST cover from \Cref{cor:HSTcoverGeneralMetric} we obtain:
\begin{corollary}\label{thm:general}
	For any parameters $\nu \in (0, 1/6)$, $t \in \N$, $\eps \in (0, 1/2)$, any metric space admits an oblivious $\nu$-reliable $(12t + \eps)$-spanner with size $\tilde{O}\left(n^{1+1/t}\cdot\nu^{-2}\cdot\eps^{-3}\right)$ and lightness $n^{1/t}\cdot\tilde{O}(\nu^{-2}\cdot\eps^{-4})\cdot\polylog(n)$.
\end{corollary}
\begin{proof}
	We can assume that $t\le \log n$, as taking larger $t$ will not reduce size or lightness.
	Using \Cref{thm:ultrametric_cover_to_reliable_spanner} on the $k$-HST cover from \Cref{cor:HSTcoverGeneralMetric}, we obtain an oblivious $\nu$-reliable spanner with stretch $(2+\frac{2}{200\cdot t^{3}/\eps})\cdot(2t+\eps)\le4t+3\eps$, size 
	\[
	n\cdot \tilde{O}\left(\left(n^{1/t}\cdot\log n\cdot\tilde{O}(t^{2}/\eps)\right)^{3}\cdot(\nu^{-1}\cdot\log\log n)^2\right)=\tilde{O}\left(n^{1+3/t}\cdot\nu^{-2}\cdot\eps^{-3}\right)~.
	\]	
	and lightness
	\[
	\eps^{-1}\cdot t^{3}\cdot\log n\cdot\tilde{O}\left(\left(n^{1/t}\cdot\log n\cdot\tilde{O}(t^{2}/\eps)\right)^{3}\cdot(\nu^{-1}\cdot\log\log n)^2\right)=n^{3/t}\cdot\tilde{O}(\nu^{-2}\cdot\eps^{-4})\cdot\polylog(n)~,
	\]
	where in the equality we assumed $\nu,\eps\ge\frac1n$ (as trivially every spanner has size and lightness $O(n^2)$).
	The corollary follows by replacing $t$ with $3t$ (and scaling $\eps$ accordingly).	
\end{proof}

For stretch $t=\log n$, the lightness of \Cref{thm:general} is $\approx\nu^{-2}\cdot\polylog(n)$, while by \Cref{thm:PathObliviousLB}, $\Omega(\nu^{-2}\cdot\log n)$ lightness is necessary (even for preserving only the connectivity of the path metric).
In \Cref{sec:logNstretchSpanner} (see \Cref{cor:HSTcoverGeneralMetricLogn}) we construct a light reliable $O(\log n)$-spanner with lightness $\tilde{O}(\nu^{-2}\cdot\log^4n)$.

\subsection{Minor Free Graphs}
In this subsection we use \Cref{cor:HSTcoverGeneralMetric} to obtain a reliable $(4+\eps)$-spanner for minor free graphs. Later, in \Cref{thm:MinorFreeOptimalStretch} we will improve the stretch to a near optimal $2+\eps$. Nevertheless, if the goal is to minimize lightness, the result in this subsection is better.
By applying \Cref{thm:pairwise_partition_cover_to_ultrametric_cover} on the PCSS of \Cref{thm:minor_free:pairwise} 
we conclude
\begin{corollary}\label{cor:HSTcoverMinorFree}
	\sloppy Let $G$ be an $n$-vertex graph excluding a fixed minor. For any $\eps \in (0,1/12)$, $G$  admits a $O(\frac{\log n}{\eps})$-light $\left(\log n\cdot\tilde{O}(\eps^{-2}),2+\eps\right)$-$\frac{32}{\eps}$-HST cover.
\end{corollary}
\begin{proof}
	Fix $k=\frac{32}{\eps}$, and apply \Cref{thm:pairwise_partition_cover_to_ultrametric_cover} on the PCSS of \Cref{thm:minor_free:pairwise}. As a result we obtain a  $O(\frac{\log n}{\eps})$-light 
	$(\tau,\rho)$-$k$-HST cover for
	\begin{align*}
		\tau & =O(\frac{\eps^{-1}\log n}{\eps}\log k)=\log n\cdot\tilde{O}(\eps^{-2})\\
		\rho & =\frac{2}{1-6\eps}\cdot(1+3\eps)=2+O(\eps)~.
	\end{align*}
	The corollary follows by rescaling $\eps$ accordingly (and noting that it will still be $\frac{32}{\eps}$-HST cover).
\end{proof}

By applying \Cref{thm:ultrametric_cover_to_reliable_spanner} on the HST cover from \Cref{cor:HSTcoverMinorFree} we obtain:
\begin{corollary}\label{thm:MinorFreeStretch4}
	Let $G$ be an $n$-vertex graph excluding a fixed minor. For any $\eps,\nu \in (0,1/20)$, $G$  admits  an oblivious $\nu$-reliable $(4 + \eps)$-spanner with size $\tilde{O}\left(n\cdot\eps^{-6}\cdot\nu^{-2}\right)$ and lightness $\tilde{O}(\eps^{-7}\cdot\log^{4}n\cdot\nu^{-2})$.
\end{corollary}
\begin{proof}
	Using \Cref{thm:ultrametric_cover_to_reliable_spanner} upon the $k$-HST cover from \Cref{cor:HSTcoverMinorFree}, we obtain a $\nu$-reliable spanner with stretch $(2+\frac{10}{32/\eps})\cdot(2+\eps)<4+3\eps$, size $n\cdot \tilde{O}\left(\left(\frac{\log n}{\eps^{2}}\right)^{3}\cdot\nu^{-2}\right)=\tilde{O}\left(n\cdot\eps^{-6}\cdot\nu^{-2}\right)$, and lightness
	$O(\frac{\log n}{\eps})\cdot\tilde{O}((\eps^{-2}\cdot\log n)^{3}\cdot\nu^{-2})=\tilde{O}(\eps^{-7}\cdot\log^{4}n\cdot\nu^{-2})$.	
	The corollary follows by rescaling $\eps$ accordingly.
\end{proof}

\subsection{Doubling Metric of High Dimension}
Consider a metric space with a moderately large doubling dimension $\ddim$, e.g. $\sqrt{\log n}$. The reliable spanner from \Cref{thm:doubling} has exponential dependence on the dimension in both size and lightness, which might be too large. 
Nevertheless, such a metric space is much more structured than a general metric space (that has doubling dimension $O(\log n)$), and thus we expect to be able to construct better spanners for such graphs (compared to \Cref{thm:general}).
Such a phenomena was previously shown for light spanners \cite{FN22}, and for reliable sparse spanners \cite{Fil23}.
We begin by observing that a PPCS for such metric spaces follow by the sparse covers of Filtser \cite{Fil19padded}.

\begin{lemma}[\cite{Fil19padded} implicit]\label{lem:pairwise-partitionlarge-ddim}
	\sloppy  Every $n$-point metric space $(X,d_X)$ with doubling dimension $\ddim$ admits a
	$(2^{O(\frac{\sddim}{t})}\cdot\ddim\cdot t, t,\frac1t)$-PPCS, for any $\Omega(1)\le t\le\ddim$.
\end{lemma}
\begin{proof}
	Fix the scale parameter $\Delta>0$. 
	Filtser \cite{Fil19padded} constructed a collection $\mathbb{P} = \{\mathcal{P}_{1},\dots,\mathcal{P}_{s}\}$ of $s=2^{O(\frac{\sddim}{t})}\cdot\ddim\cdot t$, $\Delta$-bounded partitions, such that every ball of radius $R=\frac{2}{t}\cdot\Delta$ is fully contained in some cluster, in one of the partitions.
	We argue that $\mathbb{P}$ is an $(2^{O(\frac{\sddim}{t})}\cdot\ddim\cdot t, t,\frac1t)$-PPCS.
	
	Consider two points $x,y$ such that $d_X(x,y)\le \frac12R=\frac{\Delta}{t}$. There is some partition $\cP_i\in\mathbb{P}$, and a cluster $C\in\cP_i$ such that $B_X(x,R)\subseteq C$.
	For every point $z\in B_X(y,\frac12R)$, it holds that $d_X(x,z)\le d_X(x,y)+d_X(y,z)\le\frac12R+\frac12R=R$, implying $z\in B_X(x,R)$, and in particular $B_X(y,\frac12\cdot R)\subseteq C$. Similarly $B_X(x,\frac12\cdot R)\subseteq C$. It follows that $\mathbb{P}$ is a $(2^{O(\frac{\sddim}{t})}\cdot\ddim\cdot t, t,\frac1t)$-PPCS as required.
\end{proof}

By applying \Cref{thm:pairwise_partition_cover_to_ultrametric_cover}, we conclude
\begin{corollary}\label{cor:HSTcoverLargeDdim}
	Every $n$-point metric space $(X,d_X)$ with doubling dimension $\ddim$ admits an 
	$O(t^2\cdot\log n)$-light $\left(2^{O(\frac{\sddim}{t})}\cdot\ddim\cdot t\cdot\log t,t\right)$-$\frac{t^2}{2}$-HST cover, for any $\Omega(1)\le t\le\ddim$.
\end{corollary}
\begin{proof}
	Fix $k=8t^2$, $\eps=\frac{1}{12}$, and apply \Cref{thm:pairwise_partition_cover_to_ultrametric_cover} on the PCSS of \Cref{lem:pairwise-partitionlarge-ddim}. As a result we obtain a  $O(t^2\cdot\log n)$-light 
	$(\tau,\rho)$-$k$-HST cover for
	\begin{align*}
		\tau & =O(\frac{2^{O(\frac{\sddim}{t})}\cdot\ddim\cdot t}{\eps}\cdot\log k)=2^{O(\frac{\sddim}{t})}\cdot\ddim\cdot t\cdot\log t\\
		\rho & =\frac{2}{1-6\eps}\cdot t=4t~.
	\end{align*}
	The corollary follows by rescaling $t$ accordingly.
\end{proof}

By applying \Cref{thm:ultrametric_cover_to_reliable_spanner} on the HST cover from \Cref{cor:HSTcoverLargeDdim} we obtain:
\begin{corollary}\label{thm:DdimLarge}
	Every $n$-point metric space $(X,d_X)$ with doubling dimension $\ddim$ admits an  oblivious $\nu$-reliable $t$-spanner with size $n\cdot\tilde{O}\left(\nu^{-2}\right)\cdot2^{O(\frac{\sddim}{t})}\cdot\poly(\ddim,\log\log n)$ and lightness $2^{O(\frac{\sddim}{t})}\cdot\tilde{O}(\log n\cdot\nu^{-2})\cdot\poly(\ddim)$, for any $\Omega(1)\le t\le\ddim$.
\end{corollary}
\begin{proof}
	Using \Cref{thm:ultrametric_cover_to_reliable_spanner} upon the $k$-HST cover from \Cref{cor:HSTcoverLargeDdim}, we obtain a $\nu$-reliable spanner with stretch $(2+\frac{20}{t^{2}})\cdot t$, size $n\cdot\tilde{O}\left(\nu^{-2}\right)\cdot2^{O(\frac{\sddim}{t})}\cdot\poly(\ddim,\log\log n)$
, and lightness $2^{O(\frac{\sddim}{t})}\cdot\tilde{O}(\log n\cdot\nu^{-2})\cdot\poly(\ddim)$.
	The corollary follows by scaling $\eps$ accordingly.
\end{proof}

A particularly interesting choice of parameters is $t=\ddim$, where we will get an oblivious $\nu$-reliable $\ddim$-spanner of size $n\cdot\tilde{O}\left(\nu^{-2}\right)\cdot\poly(\ddim,\log\log n)$,
and lightness $\tilde{O}(\log^{2}n\cdot\nu^{-2})\cdot\poly(\ddim)$.

\subsection{General Ultrametric}
A major part of this paper is devoted to constructing light reliable spanners for $k$-HST. However, 
\Cref{thm:k-HST-new} requires $k>1$, and the stretch grows as $k$ is closer to 1. What about the general case of $1$-HST (a.k.a ultrametric)? A stretch of $8$ can be obtained trivially by first embedding the ultrametric into a 2-HST with distortion 2 (see \cite{BLMN03}). However, we would like preserve the near optimal stretch of $2+\varepsilon$.
In this subsection we provide an answer for this question. We begin be constructing a $k$-HST cover for ultramterics.

\begin{lemma}\label{lem:HSTcoverForUltrametric}
	For every $\eps\in(0,1)$, every ultrametric admits an $\eps^{-1}$-light $\left(O(\eps^{-1}\log\frac1\eps),1+\eps\right)$-$\frac1\eps$-HST cover.
\end{lemma}
\begin{proof}
	Consider a $1$-HST $T$. Fix $N=\left\lceil \log_{1+\eps}\frac{1}{\eps}\right\rceil=O(\eps^{-1}\log\frac1\eps)$.
	For every $i\in\{0,1,\dots,N \}$, let $T_i$ be the HST $T$, where we change the label of every internal node $x$, from $\Gamma_x$ to $(1+\eps)^i\cdot \frac{1}{\eps^j}$, for $j\in\mathbb{Z}$ such that 
	\[
	(1+\eps)^{i}\cdot\frac{1}{\eps^{j-1}}<\Gamma_{x}\le(1+\eps)^{i}\cdot\frac{1}{\eps^{j}}~.
	\]
	Finally, contract all the internal nodes that have the same label as their father. 
	As a result, we obtain a dominating $\frac{1}{\eps}$-HST $T_i$, where the distance between every two vertices is increased by at most a factor of $\frac1\eps$. In particular, $T_i$ has weight at most $\frac1\eps$ times larger than $T$.
	It remains to show that the distance between every pair of leaves is preserved up to a factor of $1+\eps$ in one of the $\frac1\eps$-HST's in the cover.
	Consider a pair $u,v$ with lca $x$, and let $i\in\{0,\dots,N\}$, $j\in\mathbb{Z}$ such that $(1+\eps)^{i-1}\cdot\frac{1}{\eps^{j}}<\Gamma_{x}\le(1+\eps)^{i}\cdot\frac{1}{\eps^{j}}$.
	In the HST $T_i$, the label of the lca of $u,v$ will be changed to $(1+\eps)^{i}\cdot\frac{1}{\eps^{j}}$, and hence $d_{T_{i}}(u,v)=(1+\eps)^{i}\cdot\frac{1}{\eps^{j}}<(1+\eps)\cdot\Gamma_{x}=(1+\eps)\cdot d_{T}(u,v)$.
\end{proof}
By applying \Cref{thm:ultrametric_cover_to_reliable_spanner} on the HST cover from \Cref{lem:HSTcoverForUltrametric} (and scaling $\eps$ accordingly) we obtain:
\begin{theorem}\label{thm:GeneralUltrametric}
	\sloppy For any parameters $\nu,\eps \in (0, 1/12)$, every ultrametric ($1$-HST) $T$ admits an oblivious $\nu$-reliable $(2 + \eps)$-spanner of size $n\cdot\tilde{O}\left(\eps^{-3}\cdot(\nu^{-1}\cdot\log\log n)^2\right)$
	and lightness $\tilde{O}(\eps^{-4}\cdot(\nu^{-1}\cdot \log\log n)^2)$.
\end{theorem}

\section{Light Reliable Spanner for the Path Graph}\label{sec:path}
In this section we present our hop-bounded oblivious reliable 1-spanner for the weighted path graph.

Let $P_n = ([n], E)$ be a weighted path on $n$ vertices and let $\nu \in (0, 1)$, $h \in [\log n]$ be two parameters of the construction. The parameter $\nu$ is the input reliablity parameter, while the parameter $h$ governs the tradeoff between the hop-bound of the spanner, to its size and lightness. 
As previous works \cite{FL22,Fil23} were concerned with the hop parameter (as in some scenarios it governs stretch), we prove \Cref{thm:path} for a general hop parameter $h$.

\begin{theorem}\label{thm:path}
    For any parameters $\nu \in (0, 1)$, and $h \in [\log n]$, any weighted path graph $P_n$ admits an oblivious $\nu$-reliable $(2h+1)$-hop $1$-spanner with lightness $O\left(h n^{2/h}\cdot \left(\frac{\log (h/\nu)}{\nu}\right)^2\right)$ and size $O\left(n^{1+1/h}\cdot\frac{\log (h/\nu)}{\nu}\right)$.
\end{theorem}

 By setting $h = \lfloor (\log n -1)/2\rfloor $, we get the following corollary:

\begin{corollary}\label{cor:path:logn}
    For any weighted path graph $P_n$, and parameter  $\nu \in (0, 1)$, there is an oblivious $\nu$-reliable,  $\log n$-hop $1$-spanner 
    with lightness $\tilde{O}(\nu^{-2}\cdot\log n)$ and size $O\left(\nu^{-1}\cdot n\cdot \log\left(\frac{\log n}{\nu}\right)\right)$.
\end{corollary}


\subsection{Construction}
Let $[n]= V_0 \supseteq V_1 \supseteq\dots \supseteq V_{h}$ be a hierarchy of randomly selected sets, such that for all $1 \le i \le h$, every vertex of $V_{i-1}$ is taken into $V_i$ independently with probability $p = n^{-1/h}$.
Let $\ell =c\cdot\nu^{-1}\cdot\ln\left(\frac h\nu\right)$ for some constant $c$ to be fixed later. Assume w.l.o.g.\ that $\ell$ is an integer.

For every index $0\le i < h$ and $x \in V_i$, let $x\le u_1<...<u_\ell=u$ be the first $\ell$ vertices of $V_{i+1}$ that lie to the right of $x$, and similarly $x\ge v_1>...>v_\ell=v$ the first $\ell$ vertices of $V_{i+1}$ that lie to the left of $x$. If there are less than $\ell$ such vertices to the right (resp., left), we simply define $u=v_n$ as the last vertex (resp., $v=v_1$ as the first vertex). Now, for every $y\in [v,u]\cap V_i$, add the edge $\{x,y\}$ to the spanner $H$. In other words, we connect $x\in V_i$ to every vertex of $V_i$ that is not farther than the first $\ell$ neighbors of $x$ in $V_{i+1}$ (in either direction).
    
Finally, vertices in $V_h$ connect to all other vertices in $V_h$.
Denote by $E_i$ the edges we added at step $i$ to the spanner.

\subsection{Analysis}

\paragraph{Size analysis.}
Take $0\le i<h$, and condition on any fixed choice of $V_i$. Consider any vertex $x\in V_i$, and arrange the vertices of $V_i$ that lie to the right of $x$ in increasing order. For each such vertex we throw an independent coin with probability $p$ for success (meaning it goes to $V_{i+1}$ with this probability). Note that the number of edges $x$ adds to the right in step $i$ is essentially the number of coins we throw until the $\ell$-th success. (In fact, the number of edges can only be smaller if there are less than $\ell$ successes when we run out of vertices in $V_i$.) The expected number of trials until we see $\ell$ successes is $\ell/p$. The same argument holds for the left side edges. 

This bound holds for any choice of $V_i$. Note that for $0\le i\le h$, $E[|V_i|] = n p^i$, so the expected number of edges added in step $i$ for $0\le i<h$ is at most
\[
n p^i\cdot 2\ell/p = 2n p^{i-1}\cdot \ell~,
\]
and over the first $h$ steps it is at most
\[
2n\ell\cdot\sum_{i=0}^{h-1}p^{i-1} = O(n\ell/p) = O(n^{1+1/h}\cdot\ell)~,
\]
using that $p=n^{-1/h}\le 1/2$.
For $i=h$ we add at most $|V_h|^2$ edges. In expectation:
\begin{equation}
\E[|V_{h}|^{2}]=\sum_{i}\Pr[v_{i}\in V_{h}]+\sum_{i\ne j}\Pr[v_{i},v_{j}\in V_{h}]=n\cdot p^{h}+n\cdot(n-1)\cdot p^{2h}<2~.\label{eq:edgesEh}
\end{equation}
We conclude that the expected size of the spanner is $O\left(n^{1+1/h}\cdot\ell\right)$.
\paragraph{Lightness Analysis.}
Fix any edge $\{u,v\}\in E(P_n)$, we say that a spanner edge $\{x,y\}$ {\em crosses} the edge $\{u,v\}$ if $x\le u$ and $v\le y$. Let $c(u,v)$ denote the number of times $\{u,v\}$ is crossed. Observe that the weight of each spanner edge is equal to the sum of weights of edges in $P_n$ that it crosses, therefore, the total weight of the spanner is
\[
\sum_{e\in P(n)}c(e)\cdot w(e)~.
\]
Thus, it suffices to show that for every edge $e\in E(P_n)$: $$\E[c(e)]\le O(h n^{2/h}\cdot\ell^2)~.$$ 

To this end, fix an edge $\{u,v\}\in E(P_n)$, and an index $0\le i<h$. We will bound the expected number of edges in $E_i$ that cross $\{u,v\}$. Condition on any fixed choice of $V_i$. Note that an edge $\{x,y\}$ with $x,y\in V_i$, $x\le u$ and $y\ge v$ is added to $E_i$ by $x$ iff there are less than $\ell$ vertices of $V_{i+1}$ in the interval $[x:y)$.

Consider the vertices of $V_i$ from $u$ to the left in decreasing order, and similarly to the above lemma, let $X$ be a random variable counting the number of coins (with probability $p$ for success) we throw until getting $\ell$ successes. Denote by $Y$ the symmetric random variable, when considering vertices of $V_i$ from $v$ to the right, in increasing order. Then observe that at most $X\cdot Y$ edges of $E_i$ cross $\{u,v\}$. Since $X,Y$ are independent, we have that
\[
\E[X\cdot Y] = \E[X]\cdot \E[Y]\le(\ell/p)^2~.
\]
By (\ref{eq:edgesEh}),  the expected number of edges in $E_h$ is bounded by $2$, so each edge of $P_n$ is expected to be crossed at most twice by edges in $E_h$.
 Overall, when considering all the $h+1$ levels, for each $e\in E(P_n)$
\[
\E[c(e)]\le O(h\cdot\ell^2/p^2) = O\left(h\cdot n^{2/h}\cdot \ell^2\right)~,
\]
We conclude that the expected lightness of the spanner is $O\left(h\cdot n^{2/h} \cdot\ell^2\right)$.

\paragraph{Stretch and hop-bound analysis.}
We say a path $p = (v_0, \dots, v_k)$ is monotone if it is either monotone increasing: $v_0 \le \dots \le v_k$, or monotone decreasing: $v_0 \ge \dots \ge v_k$.
The following definition is crucial for our analysis of which vertices survive an attack $B$, and which will be added to $B^+$.

\begin{definition}
    We say a monotone increasing (resp.\ decreasing) path $p = (v_0, \dots, v_k)$ of the spanner $H$ is \emph{usable} for $v_0$ if the following holds.
    \begin{enumerate}
        \item For every $0\le i\le k$, $v_i \in V_i$.\label{enum:usable:in_vi}
        \item For every  $0\le i< k$, if $v_i \neq v_{i+1}$, then $\{v_i, v_{i+1}\} \in E_i$.\label{enum:usable:from_e_i}
        \item $v_k$ is connected in $H$ to all vertices in $V_k \cap [v_k:n]$ (resp.\ $V_k \cap [1:v_k]$)\label{enum:usable:connected}
    \end{enumerate}

    We say a vertex $v$ is \emph{safe} w.r.t. an attack $B\subseteq V$, if it has a monotone increasing \emph{usable} path and a monotone decreasing \emph{usable} path which are both disjoint from the attack $B$.
\end{definition}

The following lemma asserts that the spanner contains a shortest path that is not damaged by the attack (also with a bounded number of hops) between safe vertices.

\begin{lemma}\label{lem:path:safe_connected}
    If $u, v \in [n]$ are \emph{safe} w.r.t. an attack $B$, then the spanner contains a $(2h+1)$-hop monotone path between $u,v$ that is disjoint from $B$.
\end{lemma}
\begin{proof}
 
    Assume w.l.o.g.\ that $u < v$ and let $(u=u_0, \dots, u_k)$ be a \emph{usable} monotone increasing path of $u$ and $(v=v_0, \dots, v_j)$ a monotone decreasing \emph{usable} path of $v$. Additionally, assume w.l.o.g.\ that $k \le j$. 

    If $u_k \le v_k$, then by \cref{enum:usable:connected}, $u_k$ is connected to every vertex in $[u_k:n]\cap V_k$, in particular the spanner contains the edge $\{u_k,v_k\}$. Thus, we may take the monotone path $u_0, \dots, u_k, v_k, \dots, v_0$.

    Otherwise, there exists $i < k$ s.t. $u_i < v_i$ and $u_{i+1} \ge v_{i+1}$. Recall that by our spanner construction, $u_i$ is also connected to all the vertices $[u_i:u_{i+1}] \cap V_i$, and $v_i$ is connected to all the vertices $[v_{i+1}:v_i] \cap V_i$.
    If $v_i\le u_{i+1}$ then $v_i \in [u_i:u_{i+1}]$, and we may use the monotone path $u_0, \dots, u_i, v_{i}, \dots, v_0$.
    Else, $u_{i+1} < v_i$, therefore $u_{i+1} \in [v_{i+1}:v_i]$, and as $u_{i+1}\in V_i$ as well, we have the motonote path $u_0, \dots, u_{i+1}, v_{i}, \dots, v_0$.

    It remains to bound the number of hops. Note that by \cref{enum:usable:in_vi}, a \emph{usable} path contains at most $h$ edges, and every $u-v$ path we considered here is a concatenation of (a prefix of) two such paths, so the number of edges used is at most $2h+1$.
 \end{proof}

\paragraph{Reliability analysis.}
Let $B$ be an oblivious attack. For any spanner $H$ in the support of the distribution, the faulty extension $B^+:=B_H^{+}$ will consist of $B$ and all the vertices $v$ that are not {\em safe}. Recall that the attack is oblivious to our choice of the random sets $V_i$. In the remainder of this section, for each vertex we analyse the probability that it is safe, which will depend on the number of faulty vertices in its neighborhoods, as captured by the notion of {\em shadow}.

\begin{definition}[\cite{BHO19}]\label{def:shaddow}
    Let $P_n$ be a path graph and let $B$ be a subset of its vertices $(B \subseteq [n])$.
    The \emph{left $\alpha$-shadow} of $B$ is all the vertices $b$ such for some $a \in [n], a\le b$, $|[a:b]\cap B|\ge\alpha\cdot|[a:b]|$, denoted by $\mathcal{S}_L(\alpha,B)$. The \emph{right $\alpha$-shadow} $\mathcal{S}_R(\alpha,B)$ is defined symmetrically. The set  $\mathcal{S}_\alpha(B)=\mathcal{S}_L(\alpha,B)\cup \mathcal{S}_R(\alpha,B)$ is called the \emph{$\alpha$-shadow} of $B$.
    If $B$ is clear from context, we may simply write $\mathcal{S}_\alpha$ for the \emph{$\alpha$-shadow} of $B$.
\end{definition}

\begin{lemma}[\cite{BHO19}]\label{lem:small_shadow}
For any $B\subseteq [n]$:
\begin{itemize}
\item   For every $\alpha\in [\frac23,1)$ , $|\mathcal{S}_\alpha|\le \frac{|B|}{2\alpha-1}$.

\item    For every $\alpha\in (0,1)$, $|\mathcal{S}_\alpha|\le O\left(\frac{|B|}{\alpha}\right)$.
\end{itemize}
\end{lemma}

The following lemma provides a quantitative bound, exponential in the parameter $\ell$, on the failure probability of vertices outside a certain shadow.

\begin{lemma}\label{lem:not-safe}
For any $0<\alpha<1$, if $x\in [n]\setminus S_\alpha$, then 
\[
\Pr[x\textrm{ is not safe}]\le O(\sqrt{\ell}\cdot h\cdot\alpha^{\ell-1})~.
\]
\end{lemma}
\begin{proof}
Note that $x\notin B$, as otherwise by definition it will be contained in $S_\alpha$ for any $0\le\alpha\le 1$. We will try to construct a usable monotone increasing path for $x$, $(v_0,v_1,...,v_k)$ for some $0\le k\le h$, that is disjoint from $B$. Initially set $v_0=x\in V_0\setminus B$. Assume we built the path until $v_i\in V_i\setminus B$, and now we attempt to find the next vertex $v_{i+1}\in V_{i+1}$.

Consider the first $\ell$ vertices in $V_{i+1}$ that lie to the right of $x$.
If there are less than $\ell$ such neighbors, then observe that there are less than $\ell$ vertices in $V_{i+1}$ to the right of $v_i$ as well (as $v_i\ge x$). In this case, by the spanner construction, $v_i$ connects to all vertices in $V_i$ to its right, and we can set $k=i$ and stop the process (observe that $v_k$ will satisfy item \ref{enum:usable:connected} in the definition of usable path, so indeed we may stop here). Otherwise, if there is a vertex in $V_{i+1}\setminus B$ among the first $\ell$ neighbors of $x$, we may take the first such vertex as $v_{i+1}$. Note that the path remains monotone: $v_i\le v_{i+1}$. This is because $v_{i+1}\in V_i$, i.e. it was a valid choice for $v_i$, and we always take the first possible vertex.

We conclude that the only case the path-building fails is the event that all these $\ell$ vertices in $V_{i+1}$ fall in $B$. 

By the virtue of $x\notin S_\alpha$, we have that in any interval $[x:y]$ (for $y>x$), at most $\alpha$ fraction of the vertices are in $B$. Fix any $y>x$, and condition on the event that $y$ is the smallest such that the first $\ell$ neighbors in $V_{i+1}$ to the right of $x$ are in the interval $I=[x:y]$. Recall that every vertex is sampled to $V_{i+1}$ obliviously to the attack $B$. Note that the conditioning does create dependencies and change the probability to be in $V_{i+1}$, but the main observation is, that except for the vertex $y\in V_{i+1}$, every set of $\ell-1$ vertices in $[x:y)$ has equal probability to be the remaining $\ell-1$ vertices of $V_{i+1}$. Thus, the failure probability at step $i+1$, which is the probability that these $\ell$ vertices in $V_{i+1}$ are all taken from the set $B$, is at most
\begin{equation}\label{eq:binom}
\frac{\binom{|I\cap  B|}{\ell-1}}{\binom{|I|}{\ell-1}}\le \frac{\binom{\alpha|I|}{\ell-1}}{\binom{|I|}{\ell-1}}\le O(\sqrt{\ell}\cdot \alpha^{\ell-1})~.
\end{equation}
The last inequality uses standard approximation of binomial coefficients, see \appendixref{app:sterling} for a proof. The lemma follows by noticing that the bound obtained is independent of $y$, and by taking a union bound over both sides (left and right) of the at most $h$ steps $i=0,1,...,h-1$.
\end{proof}

We will consider two regimes of shadows separately, the first when $\alpha$ is close to 1, and the second for small $\alpha$. For the first regime, define for each index $0\le j\le\lfloor\log\frac{1}{3\nu}\rfloor$, $\alpha_j=1-2^j\cdot\nu$. Note that for any such $j$, $\alpha_j\ge 2/3$, so by the first item in \Cref{lem:small_shadow} we have 

\[
|S_{\alpha_j}|\le \frac{|B|}{2\alpha_j-1}= \frac{|B|}{1-2^{j+1}\nu}\le (1+2^{j+2}\nu)|B|~.
\]
Since all vertices of $B$ are included in any shadow, it follows that
\begin{equation}\label{eq:Shadow-large}
|S_{\alpha_j}\setminus B|\le 2^{j+2}\nu|B|~.
\end{equation}
For the smaller shadows, by the second item in \Cref{lem:small_shadow} we have 
\begin{equation}\label{eq:Shadow-small}
|S_{2^{-j}}|\le O(2^j|B|)~.
\end{equation}

\begin{lemma}\label{lem:B+}
$\E[|B^+|]\le (1+O(\nu))|B|$.
\end{lemma}
\begin{proof}
First, consider the case that $B=\emptyset$. Note that in this case, every vertex is safe, as it has a monotone increasing and a monotone decreasing usable paths. To see the former: for $0 \le i < h$, every vertex $v \in V_i$ is either connected to the closest vertex of $V_{i+1}$ that lie to the right of $v$, or, if there is no such vertex, then $v$ is connected to every vertex in $V_i \cap [v:n]$. Thus one can easily build a monotone increasing path. Therefore, in this case $B^+ = B=\emptyset$.

Notice that
\begin{align}\label{eq:bplus_sum_unasfe}
\E[|B^+|]\le|B|+\sum_{x\in [n]\setminus B}\Pr[x\textrm{ is not safe}]~.
\end{align}

We analyze \Cref{eq:bplus_sum_unasfe} by considering vertices in different shadow regimes separately, i.e.,
\[
[n] = S_{\alpha_0}+\sum_{j=1}^{\lfloor\log\frac{1}{3\nu}\rfloor}\left(S_{\alpha_j}\setminus S_{\alpha_{j-1}}\right) + S_{1/2}\setminus S_{\alpha_{\lfloor\log\frac{1}{3\nu}\rfloor}} + \sum_{j=2}^{\log n}S_{2^{-j}}\setminus S_{2^{-(j-1)}}~.
\]
Note that $S_{1/n}=[n]$, as $B\neq\emptyset$, so every vertex was accounted for.

It holds that
\begin{align*}
\mathbb{E}\left[|B^{+}|\right]\le & \underset{(1)}{\underbrace{|\mathcal{S}_{\alpha_{0}}|}}+\underset{(2)}{\underbrace{\sum_{j=1}^{\log\frac{1}{3\nu}}\sum_{x\in\mathcal{S}_{\alpha_{j}}\setminus\mathcal{S}_{\alpha_{j-1}}}\Pr\left[x\in B^{+}\right]}}\\
 & \quad+\underset{(3)}{\underbrace{\sum_{x\in\mathcal{S}_{\frac{1}{2}}\setminus\mathcal{S}_{\alpha_{\lfloor\log\frac{1}{3\nu}\rfloor}}}\Pr\left[x\in B^{+}\right]}}+\underset{(4)}{\underbrace{\sum_{j=2}^{\log n}\sum_{x\in\mathcal{S}_{2^{-j}}\setminus\mathcal{S}_{2^{-(j-1)}}}\Pr\left[x\in B^{+}\right]}}~.
\end{align*}
We next bound each one of the summands:\footnote{For convenience we will ignore the $-1$ in the exponent of $\alpha$ in \cref{lem:not-safe}, it can easily be handled by increasing slightly $\ell$.}
\begin{enumerate}
\item By \Cref{eq:Shadow-large}, $(1)=|S_{\alpha_0}|\le (1+4\nu)\cdot |B|$.
\item Fix $1\le j\le\lfloor\log\frac{1}{3\nu}\rfloor$, and $x\notin  S_{\alpha_{j-1}}$, then by \Cref{lem:not-safe} the probability that $x$ is not safe is at most
\[
O(\sqrt{\ell}\cdot h)\cdot(1-2^{j-1}\nu)^{c\cdot\nu^{-1}\cdot\ln(h/\nu)}\le O(h/\nu)^{2}\cdot e^{-2^{j-1}\cdot c\cdot\ln(h/\nu)}\le2^{-2j}~,
\]
where the last inequality holds for large enough constant $c$.
 By \Cref{eq:Shadow-large}, $|S_{\alpha_j}\setminus B|\le 4\nu\cdot 2^j |B|$. Summing over all indices $j$ we conclude $(2)\le\sum_{j=1}^{\log\frac{1}{3\nu}}4\nu\cdot2^{j}|B|\cdot2^{-2j}\le4\nu\cdot|B|$.

\item
For the transition between large and small shadows, whenever $x\in S_{1/2}\setminus S_{\alpha_{\lfloor\log\frac{1}{3\nu}\rfloor}}$, since $\alpha_{\lfloor\log\frac{1}{3\nu}\rfloor}\le 5/6$ we have that the probability that $x$
is not safe is at most 
\[
O(h/\nu)^{2}\cdot(5/6)^{-c\cdot\nu^{-1}\cdot\ln(h/\nu)}\le\nu~,
\]
for large enough $c$.
By \Cref{eq:Shadow-small}, $|\mathcal{S}_{\frac{1}{2}}\setminus B|\le O(|B|)$, thus $(3)\le O(\nu|B|)$.

\item
For $2\le j \le \log n$ and $x\notin S_{2^{-(j-1)}}$, by \Cref{lem:not-safe} the probability that $x$ is not safe is at most
\[
O(\sqrt{\ell}h)\cdot(2^{-(j-1)})^{c\cdot\nu^{-1}\cdot\ln(h/\nu)}\le O(h/\nu)^{2}\cdot(\nu/h)^{j\cdot c}\le2^{-2j}\cdot\nu~,
\]
for large enough constant $c$.
By \Cref{eq:Shadow-small}, $|S_{2^{-j}}|\le O(2^j|B|)$. It follows that $(4)\le\sum_{j=2}^{\log n}O(2^{j}|B|)\cdot2^{-2j}\cdot\nu=O(\nu)\cdot|B|$.

\end{enumerate}
Combining the $4$ cases together, we conclude that $\mathbb{E}\left[B^{+}\right]\le(1+O(\nu))\cdot|B|$,  as required.
\end{proof}

\begin{proof}[Proof of \Cref{thm:path}]
    The bounds on the expected size and lightness of the spanner were shown above, and by \Cref{lem:expect-the-worst}, they can be translated to worst-case bounds, incurring only a constant loss.

    Recall that we set $B^+$ to be all the vertices which are not \emph{safe}. By \Cref{lem:path:safe_connected} we get a shortest path with $2h+1$ hops for any pair of \emph{safe} vertices.
 By \Cref{lem:B+}, the expected size of $B^+$ is $(1+O(\nu))|B|$, the theorem follows by rescaling $\nu$ by a constant.
\end{proof}

\section{Improved Light Reliable Spanners for Minor-free Graphs}\label{sec:minor-free}

In this section we refine our techniques in order to obtain near optimal stretch for light reliable spanners of minor-free graphs. More generally, we show that a certain property of the Pairwise Partition Cover Scheme (PPCS) allows us to improve the stretch to be almost $2$, which is near optimal, while increasing the lightness by polylog factors. We begin by formally defining this property, which could be useful for other graph families as well. Throughout this section $G=(X,E,w)$ is a weighted graph with $n$ vertices excluding a constant size minor. $d_G$ denotes the shortest path metric in $G$. That is $d_G(u,v)$ denotes the minimum weight of a path from $u$ to $v$ in $G$.

\paragraph{Centrally-padded PPCS for Minor-free Graphs.}
The property of PPCS we will exploit is captured by the following definition.

\begin{definition}
	A $(\tau,\rho,\eps,\Delta)$-pairwise partition cover $\mathbb{P} = \{\mathcal{P}_{1},\dots,\mathcal{P}_{s}\}$ of a metric space $(X,d)$ is called {\em centrally-padded}, if every cluster $C$ in every partition has a designated center $x\in X$, and for every pair $u,v$ such that $\frac{\Delta}{2\rho}\le d_G(u,v)\le\frac{\Delta}{\rho}$, there is a cluster $C$ in one of the partitions $\mathcal{P}_{i}$ such that $C$ contains both closed balls $B(u,\eps\Delta),B(v,\eps\Delta)$, and also 
	\begin{equation}\label{eq:centrally-padded}
		d_G(u,x)+d_G(v,x)\le (1+32\epsilon)\cdot d_G(u,v)~.
	\end{equation}
\end{definition}

The following lemma asserts that our construction of PPCS for minor-free graphs in \Cref{sec:ppcs-minor-free} is in fact centrally-padded.

\begin{lemma}
	For any minor-free graph $G$ with $n$ vertices and $0<\epsilon<1/12$, there exists $\left(O(\varepsilon^{-1}\log n), \frac{2}{1 - 6\varepsilon}, \varepsilon\right)$-PPCS which is centrally-padded.
\end{lemma}
\begin{proof}
	Consider the construction of \Cref{lem:pairwise_spddepth}. Recall that every cluster is a ball centered at a net point, so we naturally define its center as that net point. For any $u,v\in X$ with $\frac{(1-6\epsilon)\Delta}{4}\le d(u,v)\le \frac{(1-6\epsilon)\Delta}{2}$, we found the first shortest path $P$ in the SPD that intersects $P_{uv}$ (the shortest $u-v$ path) or at least one of the balls $B_u=B(u,\epsilon\Delta)$, $B_v=B(v,\epsilon\Delta)$ (see \Cref{fig:minor_free:pairwise}). We denoted $x\in P$ as a vertex on that intersection. Then we found a net-point $z\in {\cal N}$ on $P$ at distance at most $\epsilon\Delta$ from $x$, and consider the cluster $C=B(z,\Delta/2)$. 
	
	If $x\in P_{uv}$ then 
	\[
	d(z,u)+d(z,v)\le 2d(z,x)+d(x,u)+d(x,v) \le 2\epsilon\Delta+d(u,v)\le (1+16\epsilon)\cdot d(u,v)~.
	\]
	Otherwise, w.l.o.g.\ $x\in B_u$ and we get that
	\[
	d(z,u)+d(z,v)\le d(z,u)+d(z,u)+d(u,v)\le 4\epsilon\Delta+d(u,v)\le (1+32\epsilon)\cdot d(u,v)~,
	\]
	as required.

\end{proof}

\paragraph{$k$-HST Cover.} The next step is to compute an $k$-HST cover, which is done exactly in the same manner as in \Cref{thm:pairwise_partition_cover_to_ultrametric_cover}, so we get a $O(k\log n)$-light $\left(O(\epsilon^{-2}\log n\cdot\log k),\frac{2(1+3\epsilon)}{1-6\epsilon}\right)$-$k$-HST cover. The main point is that we will use these $k$-HSTs to construct reliable spanners, but the edge weights and the stretch guarantees will be with respect to the original distances in the graph. That is, in some sense we ignore the distances induced by the $k$-HSTs, and just use their laminar structure. The property that we will use from the proof of \Cref{thm:pairwise_partition_cover_to_ultrametric_cover} is the following.
\begin{itemize}
	\item For every pair $u,v$, there exists a cluster $C$ of diameter at most $\Delta$ in the PPCS in which $u,v$ are centrally-padded, and so $C$ contains a net point. Thus, there will be a $k$-HST in the cover with an internal node $x$ and label $\Gamma_x=\Delta$ corresponding to $C$, that contains $u,v$.  
\end{itemize}
We remark that $L(x)$ is not necessarily equal to $C$, since we changed $C$ a bit before making it an internal node of the $k$-HST (to guarantee the laminar structure, and a bound on the lightness). 
The main result of this section is the following theorem.
\begin{theorem}\label{thm:MinorFreeOptimalStretch}
	Let $G=(V,E)$ be a graph with $n$ vertices that excludes a fixed minor. Then for any $0<\epsilon<1/12$ and $0<\nu'<1$, $G$ admits an oblivious $\nu'$-reliable $2(1+\epsilon)$-spanner of size $\tilde{O}\left(\frac{n}{\epsilon^6\nu'^2}\right)$ and lightness $\tilde{O}\left(\frac{\log^8 n}{\epsilon^7\nu'^2}\right)$.
\end{theorem}

\sloppy Let $k=c'/\epsilon$, for a constant $c'$ to be determined later.
We create a $O(k\log n)$-light $\left(\tau,\frac{2(1+3\epsilon)}{1-6\epsilon}\right)$-$k$-HST cover for the graph $G$, with $\tau=O(\epsilon^{-2}\log n\cdot\log k)$, as discussed above. Since we desire a $\nu'$-reliable spanner for $G$, we will use the parameter $\nu=\frac{\nu'}{5\tau}$ when devising a $\nu$-reliable spanner for each $k$-HST in the cover.

Let $T$ be one of the $k$-HSTs in the cover.
Note that every internal node of $T$ corresponds to a cluster $C$ in the centrally-padded PPCS, which have a center $x$. In a hope to avoid confusion, we will refer to $x$ both as the cluster center, and as the internal node of $T$. 

Recall that in \Cref{sec:hst} every internal node chose an arbitrary ordering on its leaves, which was used to define the preorder of $T$. Here, the order will not be arbitrary. Instead, it will be defined with respect to distances in $G$. That is, each internal node $x$ orders its children $x_1,\dots,x_t$ (each is a net point in the graph) by their distance to $x$ (in $G$). Then, let $P$ be the resulting preorder path on the leaves of $T$.

The intuition behind the sampling of the random bi-cliques, is that we want vertices ``near'' $x$ to be chosen, since the centrally-padded property gives us a better stretch guarantee going through $x$, than just $2\Gamma_x$. To this end, let $L(x)=(v_1,v_2,...,v_s)$ be the ordering given by the restriction of $P$ to $L(x)$. We sample each $v_j$ independently to be included in $Z_x$ with probability $p_j=\min\{1,\frac{c\cdot \ln n}{j\cdot\nu}\}$, for a constant $c$ to be determined later. 

The edges of the spanner are: For every internal node $x\in T$ with children $x_1,\dots,x_t$, for every $j=1,\dots,t$, we add all the edges $\{\{y,z\}~:~y\in Z_x,z\in Z_{x_j}\}$ to the spanner $H$, weighted according to the distances in the graph $G$. The final spanner will consist of the union of all $O(k\log n)$ spanners for all the $k$-HSTs in the cover.

\paragraph{Safe Leaves.}
Fix a $k$-HST $T$. Under an attack $B$, 
we say that a vertex $u$ is {\em safe} w.r.t. $B$, if for every ancestor $x\in T$ of $u$, $Z_x\setminus B$ contain a vertex $y$ such that 
\begin{equation}\label{eq:ssafe}
	d_G(x,y)\le d_G(x,u)+2\Gamma_x/k~. 
\end{equation}
In other words, we want that every ancestor $x$ of $u$ to have a surviving vertex $y$ in its sample set, which is not much farther than the distance of $u$ to the center $x$. 

Denote $B_T^+$ as all the 
vertices which are not safe in $T$ w.r.t. $B$. The final bad set is defined as $B^+=B\cup\bigcup_TB_T^+$. The following claim will be useful for bounding the size and lightness of our spanner.

\begin{claim}\label{clkaim:size}
	Fix any $T$ in the cover, then for any $x\in T$,
	\[
	\E[|Z_x|]\le O((\ln^2 n)/\nu)~.
	\]
\end{claim}
\begin{proof}
	Let $L(x)=(v_1,\dots,v_s)$, with the order induced by the restriction of $P$ to $L(x)$, then 
	\[
	\E[|Z_x|]=\sum_{j=1}^sp_j\le\frac{c\ln n}{\nu}\cdot\sum_{j=1}^s\frac1i = O\left(\frac{\ln^2 n}{\nu}\right)~.
	\]
\end{proof}

\paragraph{Size Analysis.}
Fix any tree $T$ in the cover. Denote $z_x=|Z_x|$, and note that these random variables $\{z_x\}_{x\in T}$ are independent, so $\E[z_x\cdot z_y]=\E[z_x]\cdot\E[z_y]$ whenever $x\neq y$. Using \Cref{clkaim:size},
the expected number of edges added to $H$ by the random bi-cliques is
\begin{align*}
	\E\left[\sum_{x\in T}\sum_{i=1}^{\deg(x)}z_{x}\cdot z_{x_{i}}\right] & =\sum_{x\in T}\sum_{i=1}^{\deg(x)}\E[z_{x}]\cdot\E[z_{x_{i}}].\\
	& =O(\nu^{-2}\cdot\log^{4}n)\cdot\sum_{x\in T}\deg(x)=O(\nu^{-2}\cdot n\log^{4}n)~.
\end{align*}

The final spanner is a union of $\tau=O(\epsilon^{-2}\log n\cdot\log k)$ spanners for each $T$ in the cover, and $\nu'=\nu\cdot\tau$, so the final size is
$\tau\cdot O((\frac{\tau}{\nu'})^{2}\cdot n\log^{4}n)=n\cdot\nu'^{-2}\cdot\epsilon^{-6}\cdot\log^{7}n\cdot\log^{3}k=\nu'^{-2}\cdot\epsilon^{-6}\cdot\tilde{O}(n)$.

\paragraph{Lightness Analysis.}

Let $T$ be any $k$-HST in the cover, and recall that the MST weight of $T$ is equal to
\[
\sum_{x\in T}(\deg(x)-1)\cdot\Gamma_x~.
\]
Each edge that $x$ adds to the spanner has weight at most $\Gamma_x$ (even though we use the graph distance, as $T$ is dominating). Using \Cref{clkaim:size}, the total weight of edges in the random bi-cliques is expected to be at most
\begin{eqnarray*}
	\E\left[\sum_{x\in T}\sum_{i=1}^{\deg(x)}\Gamma_x\cdot z_x\cdot z_{x_i}\right]&=&\sum_{x\in T}\Gamma_x\sum_{i=1}^{\deg(x)}\E[z_x]\cdot \E[z_{x_i}]\\
	&\le& O((\log^4 n)/\nu^2)\cdot\sum_{x\in T}\Gamma_x\cdot\deg(x)\\
	&=& O((\log^4 n)/\nu^2)\cdot w(MST(T))~. 
\end{eqnarray*}

Since every $k$-HST has lightness $O(k\log n)$, and there are $\tau=O(\epsilon^{-2}\log n\cdot\log k)$ trees in the cover, and $\nu'=\nu\cdot\tau$, the lightness of the resulting spanner compared to $G$ is
\begin{align*}
	\sum_{T\text{ in the cover}}O\left(\frac{\log^{4}n}{\nu^{2}}\right)\cdot\frac{w(MST(T))}{w(MST(G))} & =O\left(\frac{\tau^{3}}{\nu'^{2}}\cdot k\cdot\log^{5}n\right)\\
	& =O\left(\frac{k\cdot\log^{8}n\cdot\log^{3}k}{\nu'^{2}\cdot\eps^{6}}\right)\\
	& =\nu'^{-2}\cdot\tilde{O}(\eps^{-7}\cdot\log^{9}n)
\end{align*}

\paragraph{Reliability Analysis.}
Fix an attack $B$ and a tree $T$, and define the shadow with respect to the path $P$ and the set $B$ (recall \cref{def:shaddow}). We start by showing that for any internal node $x$, the preorder of $L(x)$ almost respects the distances to the center $x$ in $G$.

\begin{claim}\label{claim:order}
	Fix any node $x\in T$, and let $L(x)=(v_1,\dots,v_s)$ be the ordering given by the restriction of $P$ to $L(x)$. Then for any $1\le i<j\le s$ we have that
	\[
	d_G(x,v_i)\le d_G(x,v_j)+2\Gamma_x/k~.
	\]
\end{claim}
\begin{proof}
	Let $x_{i'}$ (resp., $x_{j'}$) be the child of $x$ whose subtree contains $v_i$ (resp., $v_j$). Since $v_i$ appears in $P$ before $v_j$, it follows that the order on the children of $x$ is such that $x_{i'}$ appears before $x_{j'}$ (we allow $i'=j'$). By our definition it means that $d_G(x,x_{i'})\le d_G(x,x_{j'})$. 
	
	As $T$ is a  $k$-HST, all distances in $T$ between vertices in $L(x_{i'})$ (resp., $L(x_{j'})$) are at most $\frac{\Gamma_x}{k}$. Since $T$ is dominating, this also holds for the graph distances, which gives that both $d_G(v_i,x_{i'}),d_G(v_j,x_{j'})\le\frac{\Gamma_x}{k}$. We conclude that
	\begin{align*}
		d_{G}(x,v_{i}) & \le d_{G}(x,x_{i'})+d_{G}(x_{i'},v_{i})\\
		& \le d_{G}(x,x_{j'})+d_{G}(x_{i'},v_{i})\\
		& \le d_{G}(x,v_{j})+d_{G}(x_{j'},v_{j})+d_{G}(x_{i'},v_{i})\\
		& \le d_{G}(x,v_{j})+\frac{2\cdot\Gamma_x}{k}~.
	\end{align*}
\end{proof}

\begin{lemma}\label{lem:safe-MF}
	For every tree $T$, $0<\alpha\le 1-\nu$, and any vertex $u\notin S_\alpha( B)$,  
	\[
	\Pr[u\textrm{ is not safe}]\le n^{-2}~.
	\]
\end{lemma}
\begin{proof}
	Let $x$ be any ancestor of $u$, and let  $L(x)=(v_1,\dots,u=v_j,\dots,v_s)$ be the ordering given by the restriction of $P$ to $L(x)$. If we want that $u$ will not fail to be safe due to $x$, it suffices that $Z_x\setminus B$ contains a vertex in the prefix $(v_1,\dots,v_j)$. This is because \Cref{claim:order} suggests that any such vertex will satisfy (\ref{eq:ssafe}).
	
	Since $u\notin S_\alpha( B)$, it follows that at most $\alpha$ fraction of the vertices $(v_1,\dots,v_j)$ are in $ B$. As the probability of being sampled to $Z_x$ decreases with the index, it can be easily checked that the worst possible case is that $\{v_1,\dots,v_{\lfloor\alpha\cdot j\rfloor}\}\subseteq B$ (in any other case the probability of success will only be higher). We assume that $p_{\lfloor\alpha j\rfloor+1}<1$, as otherwise $v_{\lfloor\alpha j\rfloor+1}$ is surely sampled into $Z_x$. This means $p_i=\frac{c\cdot\ln n}{i\cdot\nu}$ for all $i>\alpha\cdot j$. Note that $Z_x$ is sampled independently of $ B$, thus
	\begin{align}\label{eq:bboud}
		\Pr[Z_{x}\cap\{v_{\lfloor\alpha j\rfloor+1},\dots,v_{j}\}=\emptyset] & =\prod_{i=\lfloor\alpha\cdot j\rfloor+1}^{j}(1-p_{i})\\
		\nonumber & \le e^{-\sum_{i=\lfloor\alpha\cdot j\rfloor+1}^{j}p_{i}}\\
		\nonumber & =e^{-\sum_{i=\lfloor\alpha\cdot j\rfloor+1}^{j}\frac{c\cdot\ln n}{i\cdot\nu}}\\
		\nonumber & \le n^{-\frac{c}{2\nu}\cdot(\ln j-\ln(\alpha\cdot j))}\\
		\nonumber & =n^{-\frac{c}{2\nu}\cdot\ln(\frac{1}{\alpha})}\\
		\nonumber & \le n^{-3}~.
	\end{align}
	The last inequality holds as $\alpha\le 1-\nu$, so
	\[
	\ln(\frac{1}{\alpha})\ge\ln(\frac{1}{1-\nu})>\ln(1+\nu)\ge\frac{\nu}{2}
	\]
	and by picking a large enough constant $c\ge12$. The lemma follows by a union bound over all possible ancestors $x$.
\end{proof}

We are now ready to bound the final set $B^+$.
\begin{lemma}
	$\E[|B^+|]\le (1+\nu')\cdot|B|~.$
\end{lemma}
\begin{proof}
	First consider the case $B=\emptyset$. In this case $B^+=\emptyset$ as well. This is because all the vertices are safe. Indeed, for every tree $T$ and node $x$, the first leaf in $L(x)$ is sampled to $Z_x$ with probability $1$, and thus $Z_x\not\subseteq B$. 
	Thus we will assume $B\ne \emptyset$.
	We can also assume that $\nu'\ge \frac1n$, as otherwise \cref{thm:MinorFreeOptimalStretch} holds trivially with $H=G$.
	
	Fix a tree $T$ in the $k$-HST cover.
	As $\nu< 1/3$, by the first item of \Cref{lem:small_shadow}, the shadow $S_{1-\nu}( B)$ of the path $P$ satisfies
	\[
	|S_{1-\nu}( B)|\le\frac{| B|}{2(1-\nu)-1}\le(1+4\nu)\cdot| B|~.
	\]
	
	By \Cref{lem:safe-MF}, every vertex outside $S_{1-\nu}( B)$ joins $B^+_T$ with probability at most $n^{-2}$.
	It follows that
	\begin{align*}
		\E\left[\left|B_{T}^{+}\right|\right] & \le |S_{1-\nu}( B)\setminus B|+\sum_{v\notin S_{1-\nu}( B)}\Pr\left[v\in B_T^{+}\right]\\
		& \le4\nu\cdot| B|+n\cdot\frac{1}{n^{2}}\\
		& \le4\nu\cdot|B|+\nu\cdot|B|=5\nu\cdot|B|~.
	\end{align*}
	Summing up over all the $\tau$ trees in the cover, and recalling that $\nu=\frac{\nu'}{5\tau}$, we conclude that
	\[
	\E\left[\left|B^{+}\setminus B\right|\right]\le\sum_{T\text{ in the cover}}\E\left[\left|B_{T}^{+}\right|\right]\le\sum_{T\text{ in the cover}}5\nu\cdot|B|\le5\tau\cdot\nu\cdot|B|=\nu'\cdot|B|~.
	\]
	
\end{proof}

\paragraph{Stretch Analysis.}

Let $u,v\notin B^+$ be two safe leaves in the $k$-HST $T$ that has an internal node $x$ corresponding to a cluster $C$ in which $u,v$ are centrally-padded. Let $\Delta$ be the diameter bound on $C$, and so $\Gamma_x=\Delta$. By definition of padding we have that $u,v\in L(x)$ and
\begin{equation}\label{eq:eqeq}
	\frac{(1-6\epsilon)\Delta}{4}\le d_G(u,v)\le \frac{(1-6\epsilon)\Delta}{2}~,
\end{equation}
and as $\epsilon<1/12$, it follows that $\Delta\le 8d_G(u,v)$.

Let $x_i,x_j$ be the children of $x$ in $T$ which are the ancestors of $u$ and $v$ respectively.\footnote{Since $k=c'/\epsilon>8$, and all vertices in $L(x_i)$ are at distance at most $\Gamma_{x_i}\le\Delta/k\le 8d_G(u,v)/k$ from each other, it cannot be that $i=j$.}
As $u,v\notin B^+$, it holds that there are vertices $z\in Z_x\setminus  B$, $u'\in Z_{x_i}\setminus B$, and $v'\in Z_{x_j}\setminus B$, and by (\ref{eq:ssafe}) we also have that
\begin{equation}\label{eq:end}
	d_G(x,z)\le \min\{d_G(x,u),d_G(x,v)\}+2\Delta/k~.
\end{equation}

It follows that the edges $\{u',z\},\{v',z\}\in H$ survive the attack, and furthermore
\begin{eqnarray*}
	d_G(u',z)&\le& d_G(u',u)+d_G(u,z)\\
	&\le& \Gamma_{x_i}+d_G(u,x)+d_G(x,z)\\
	&\stackrel{(\ref{eq:end})}{\le}& \Delta/k+2d_G(u,x)+2\Delta/k\\
	&=& 2d_G(u,x)+3\Delta/k~.
\end{eqnarray*}
The second inequality uses that $u,u'\in L(x_i)$, and that $T$ is dominating. For the third one we observe that $\Gamma_{x_i} \le\Delta/k$, as $T$ is $k$-HST. A symmetric calculation shows that 
\[
d_G(v',z)\le 2d_G(v,x)+3\Delta/k~.
\]

Since we used the same bi-clique construction as we did in \cref{thm:k-HST-new}, and a more restrictive definition of {\em safe}, we have that \cref{lem:stretch-hst} still holds (with $f(x)=x$, since we did not apply the heavy-path decomposition here).
In particular, $H$ contains a $u-u'$ path (resp., $v-v'$ path) which is disjoint from $B$, of length at most $\left(1+\frac{1}{k-1}\right)\cdot\Gamma_{x_i}$ (resp.,  $\left(1+\frac{1}{k-1}\right)\cdot\Gamma_{x_j}$). As $T$ is a $k$-HST, $\Gamma_{x_i},\Gamma_{x_j}\le\frac{\Delta}{k}$, so we have that 
\begin{eqnarray*}
	d_{H\setminus B}(u,v)&\le& d_{H\setminus B}(u,u')+d_G(u',z)+d_G(z,v')+d_{H\setminus B}(v',v)\\
	&\le&2(d_G(u,x)+d_G(v,x)+3\Delta/k)+\left(2+\frac{2}{k-1}\right)\cdot \frac{\Delta}{k}\\
	&\le&2 d_G(u,v)\cdot\left(1+O(\epsilon)\right)~.
\end{eqnarray*}
The last inequality uses that $\Delta\le 8d_G(u,v)$, the definition of centrally-padded (\ref{eq:centrally-padded}), and the choice of $k=\Theta(1/\epsilon)$.

\section{Lower Bounds}\label{sec:lower}
This section is devoted to proving lower bounds. All of our lower bounds hold for any finite stretch, that is, even if one requires from the reliable spanners just to preserve connectivity. For an attack $B$ on a spanner $H$, a valid super set $B^+_H$ should satisfy that all the points in $X\setminus B^+$ belong to the same connected component in $H\setminus B$.
In \Cref{subsec:PathDeterministicLB} we show that every deterministic reliable spanner for the path must have at least $\Omega(n)$ lightness. This lower bound holds for any constant $\nu>0$.
The main point of this lower bound is that deterministic reliable spanners have huge lightness. In particular, we did not attempt to optimize the dependence on $\nu$ (or other terms). Note that deterministic $\nu$-reliable $1$-spanner with lightness $O(\nu^{-6}\cdot n\log n)$ follows from \cite{BHO19}.

In \Cref{subsec:PathObliviousLB} we prove that every {\em oblivious} $\nu$-reliable spanner for the path has lightness $\Omega(\nu^{-2}\cdot\log(n\nu))$.
In \Cref{subsec:LBHST} we construct an ultrametric such that every oblivious  $\nu$-reliable spanner has lightness $\Omega(\nu^{-2})$. 
These two lower bounds show that the lightness parameters in our \Cref{thm:path,thm:k-HST-new} are tight up to second order terms (even if we ignore the stretch factor). 
The proof of \Cref{subsec:LBHST} appears before \Cref{subsec:PathObliviousLB} as the two arguments are somewhat similar, while the proof in \Cref{subsec:LBHST} is simpler.

\subsection{Lower bound for deterministic light reliable spanners}\label{subsec:PathDeterministicLB}

\begin{theorem}\label{thm:PathDeterministicLB}[Deterministic Lower bound for Path]
	For any constant $\nu>0$, every deterministic $\nu$-reliable spanner for the unweighted path graph $P_n$ has lightness $\Omega(n)$.
\end{theorem}
\begin{proof}
	Consider a deterministic $\nu$-reliable spanner $H$ for $P_n$. Set $\epsilon = \min\{\frac{1}{16},\frac{1}{16\nu}\}=\Omega(1)$.
	Denote by $L=[1,(1/2-\epsilon)n]$ and $R=[(1/2+\epsilon)n+1,n]$ the first and last $(1/2-\epsilon)n$ vertices along $P_n$ respectively. (For simplicity we assume that these are all integers.)
	Let $E_1$ be the subset of $H$ edges going from a vertex in $L$ to a vertex in $R$. Seeking contradiction, assume that $|E_1|< \epsilon n$.
	Let $B\subseteq [n]$ be a subset consisting of all the vertices in $[(1/2-\epsilon)n+1,(1/2+\epsilon)n]$, and all the vertices in $L$ that are contained in an edge of $E_1$.
	
 Since $|B|< 3\epsilon n$, it holds that $|L\setminus B|\ge (1/2-4\epsilon)n$, and $|R\setminus B|=(1/2-\epsilon)n$.
	However, the graph $H\setminus B$ does not contain any edge from a vertex in $L\setminus B$ to a vertex in $R$. In particular, $B^+$ must contain either all the vertices in $L$, or all the vertices in $R$. It follows that $$|B^+\setminus B|\ge(1/2-4\epsilon)n\ge n/4\ge \nu\cdot 4\epsilon n>\nu\cdot|B|~,$$ 
 a contradiction to the fact that $H$ is a $\nu$-reliable spanner. 
	It follows that $|E_1|\ge \epsilon n$. Note that each edge in $E_1$ has weight at least $2\epsilon n$. We conclude that
	$$w(H)\ge |E_1|\cdot 2\epsilon n\ge 2(\epsilon n)^2=\Omega(n)\cdot w(\rm{MST})~,$$
 where in the last equality we used that $\nu$ is a constant.
\end{proof}

\subsection{Lower Bound for HST}\label{subsec:LBHST}
Similarly to \Cref{thm:PathObliviousLB}, the lower bound here holds even if one is only interested in preserving connectivity. 
\begin{restatable}[Oblivious Lower Bound for HST]{theorem}{ObliviousLowerBoundHST}
	\label{thm:HSTObliviousLB}
	For every $\nu\in(0,1)$, there is an ultrametric such that every oblivious $\nu$-reliable spanner has lightness $\Omega(\nu^{-2})$.
\end{restatable}

\hspace{-18pt}\emph{Proof.}
	Set $\ell=\frac{1}{4\nu}$. Consider an ultrametric consisting of a root $r$ with label $1$, and
		\begin{wrapfigure}{r}{0.22\textwidth}
		\begin{center}
			\vspace{-25pt}
			\includegraphics[width=0.25\textwidth]{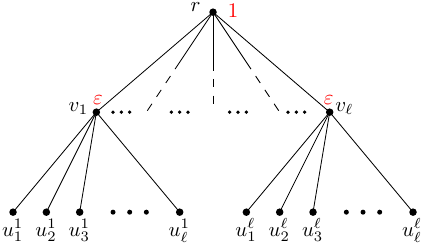}
			\vspace{-5pt}
		\end{center}
		\vspace{-23pt}
	\end{wrapfigure}
	$\ell$ children $\{v_1,\dots,v_\ell\}$, each with label $\eps=\frac1\ell$, and $\ell$ children each, where $\{v_1^i,v_2^i,\dots,v^i_\ell\}$ are the children of $v_i$. In total we have $\ell^{2}$ leaves. See illustration on the right. The MST for this ultrametric will consist of $\ell-1$ edges of weight $1$, and $\ell\cdot(\ell-1)$ edges of weight $\eps=\frac1\ell$. So the total weight is $2(\ell-1)$.
	
	Consider an oblivious $\nu$-reliable spanner $\cD$, and let $H\sim\supp(\cD)$. Let $\cJ\in\{1,2\dots,\ell\}^\ell$ be a string of $\ell$ indices between $1$ and $\ell$. Let $H_{\cJ}$ be the subgraph of $H$ induced by 
	$\{v^1_{\cJ_1},v^2_{\cJ_2},\dots,v^\ell_{\cJ_\ell}\}$. That is, for each $i$, we keep only the vertex corresponding to the $i$'th index in $\cJ$.
	
	Let $\Psi_{{\cal J}}$ be the event that the graph $H_\cJ$ contains
	at least $\frac{\ell}{2}$ edges. 
	Consider the attack $B_\cJ$ which consist of all the vertices except $\{v^1_{\cJ_1},v^2_{\cJ_2},\dots,v^\ell_{\cJ_\ell}\}$.
	If the event $\Psi_{{\cal J}}$ did not occur for a spanner $H$, then $H\setminus B_\cJ$, is disconnected, and the largest connected component has size at most $\frac\ell2$. Observe that in order to preserve connectivity, $B^+_{\cJ}$ must contain all vertices in all connected components of $H\setminus B_\cJ$, except for one component. In particular, $B^+_\cJ\setminus B_\cJ$ will contain at least $\frac\ell2=2\nu\cdot\ell^{2}\ge2\nu\cdot|B|$ vertices. As $\cD$ is $\nu$-reliable, it holds that 
	\[
	\nu\cdot|B|\ge\E[|B^{+}\setminus B|]\ge\Pr\left[\overline{\Psi_{{\cal J}}}\right]\cdot2\nu\cdot|B|~,
	\]
	It follows that $\Pr\left[\overline{\Psi_{{\cal J}}}\right]\le\frac{1}{2}$,
	and in particular $\Pr\left[\Psi_{{\cal J}}\right]\ge\frac{1}{2}$.
	We conclude that for every $\cJ$ it holds that $\mathbb{E}_{H\sim\cD}\left[\left|E(H_{\cJ})\right|\right]\ge\Pr\left[\Psi_{\cJ}\right]\cdot\frac{\ell}{2}\ge\frac{\ell}{4}$.
	
	On the other hand, denote by $\widehat{H}$ the subset of $H$ edges of weight $1$ (i.e. between children of $v_i,v_j$ for $i\ne j$).
	Note that for every $\cJ$, $H_\cJ\subseteq \widehat{H}$ (as $H_\cJ$ does not contain $\eps$-weight edges).
	Every edge $e\in \widehat{H}$ belongs to $H_\cJ$ if and only if both its endpoints are chosen by $\cJ$. If we choose $\cJ$ u.a.r., $e$ will survive with probability $\frac{1}{\ell^2}$. We conclude
	\begin{align*}
		\mathbb{E}_{{\cal J},H}\left[\left|E(H_{\cJ})\right|\right] & =\mathbb{E}_{{\cal J}}\left[\mathbb{E}_{H}\left[\left|E(H_{\cJ})\right|\right]\right]\ge\mathbb{E}_{{\cal J}}\left[\frac{\ell}{4}\right]=\frac{\ell}{4}\\
		\mathbb{E}_{{\cal J},H}\left[\left|E(H_{\cJ})\right|\right] & =\mathbb{E}_{H}\left[\mathbb{E}_{{\cal J}}\left[\left|E(H_{\cJ})\right|\right]\right]=\mathbb{E}_{H}\left[\frac{1}{\ell^{2}}\cdot|\widehat{H}|\right]~.
	\end{align*}
	As all $\widehat{H}$ edges have weight $1$,
	\[
	\mathbb{E}_{H\sim\cD}\left[w(H)\right]\ge\mathbb{E}_{H}\left[|\widehat{H}|\right]\ge\frac{\ell^{3}}{4}=\Omega(\nu^{-2})\cdot w({\rm MST)}~.
	\]
	\qed

\subsection{Lower Bound for the Unweighted Path}\label{subsec:PathObliviousLB}

In this section we prove an $\Omega(\nu^{-2}\cdot\log (n\nu))$ lower bound on the lightness any oblivious reliable spanner for the shortest path metric induced by the unweighted path (for any finite stretch parameter). As this metric has doubling dimension $1$, it follows that our light reliable spanner for doubling metrics is tight (\Cref{thm:doubling}) up to second order terms (for constant $\ddim$ and $\eps$).
\begin{restatable}[Oblivious Lower Bound for the Path]{theorem}{ObliviousLowerBoundforthePath}
	\label{thm:PathObliviousLB}
	For every $\nu\in(0,1)$, every oblivious $\nu$-reliable spanner for the unweighted path graph $P_n$ has lightness $\Omega(\nu^{-2}\cdot\log (n\nu))$.
\end{restatable}

\begin{proof}
	This proof follow similar lines to the proof of \Cref{thm:HSTObliviousLB}, however, there are some required adaptations due to the path metric, and an additional $\log n$ factor which is introduced due to the $\log n$ different scales.
	
	For simplicity we will assume that $n=2^{m}$ and $\nu=2^{-s}$
	are powers of $2$. We will also assume that 
	$m\ge s+5$.
	For every index $i$ and $H\in\supp({\cal D})$,
	denote by ${\cal E}_{H}^{i}$ the subset of $H$ edges of weight at
	least $2^{i}$.
	\begin{claim}\label{clm:PathLBSingleScale}
		For every index $i\in\left\{ s+3,\dots,m-2\right\} $, $\mathbb{E}_{H\sim{\cal D}}\left[\left|{\cal E}_{H}^{i}\right|\right]\ge\frac{1}{\nu^{2}}\cdot\frac{n}{2^{i+3}}$.
	\end{claim}
	
	\begin{proof}
		Set $\ell=\frac{1}{8\nu}$.
		Divide the path $P_{n}$ to $\frac{n}{2^{i}}$ intervals of length
		$2^{i}$.
		Remove every other interval.
		Every remaining interval, partition further into $\ell$
		intervals of length $\frac{2^{i}}{\ell}$. Denote these intervals
		by $\left\{ A_{j}^{k}\right\} _{k\in[\frac{n}{2^{i+1}}],j\in[\ell]}$
		where $$A_{j}^{k}=\left\{ v_{2(k-1)\cdot2^{i}+(j-1)\cdot\frac{2^{i}}{\ell}+1},\dots,v_{2(k-1)\cdot2^{i}+j\cdot\frac{2^{i}}{\ell}}\right\}~.$$ See illustration below.
		
		\begin{center}
			\includegraphics[width=1\textwidth]{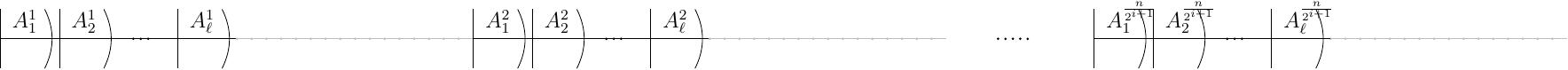}
		\end{center}

		For every subgraph $H\in\supp({\cal D})$, we create an unweighted
		supergraph $G_{H}$ where its vertex set is $\left\{ A_{j}^{k}\right\} _{k\in[\frac{n}{2^{i+1}}],j\in[\frac{1}{\ell}]}$,
		and add an edge from $A_{j}^{k}$ to $A_{j'}^{k'}$ if and only if
		$k\ne k'$ and $H$ contains an edge between points in $A_{j}^{k}$
		and $A_{j'}^{k'}$. Note that $G_{H}$ is a $\frac{n}{2^{i+1}}$-partite
		simple graph. Denote by $V(G_{H})$ and $E(G_{H})$ the vertex and
		edges sets of $G_{H}$ respectively. Clearly, every edge in $G_{H}$
		corresponds to (at least one) edge of weight at least $2^{i}$ in
		$H$. Thus, $\left|E(G_{H})\right|\le\left|{\cal E}_{H}^{i}\right|$,
		and hence in order to prove the claim it suffices to lower bound
		$\mathbb{E}_{H\sim{\cal D}}\left[\left|G_{H}\right|\right]$. We will
		proceed by a double-counting argument.
		
		Consider a $\frac{n}{2^{i+1}}$-tuple ${\cal J}=\left(j_{1},\dots,j_{\frac{n}{2^{i+1}}}\right)\in[\ell]^{\frac{n}{2^{i+1}}}$.
		The graph $G_{H}^{{\cal J}}=G_{H}\left[\left\{ A_{j_{k}}^{k}\right\} _{k\in[\frac{n}{2^{i+1}}]}\right]$
		is the induced graph by the $\frac{n}{2^{i+1}}$ vertices $\left\{ A_{j_{k}}^{k}\right\} _{k\in[\frac{n}{2^{i+1}}]}$
		of $G_{H}$. Let $B_{{\cal J}}=[n]\setminus\cup_{k}A_{j_{k}}^{k}$
		be all the vertices not in the sub intervals specified by ${\cal J}$
		(in particular $B_{{\cal J}}$ contains the $\frac{n}{2}$ vertices
		of the removed intervals). Let $\Psi_{{\cal J}}$ be an indicator for
		the event that $G_{H}^{{\cal J}}$ contains at least $\frac{n}{2^{i+2}}$
		edges. Note that if the event $\Psi_{{\cal J}}$ did not occur,
		then in $H\setminus B_\cJ$ the maximum size of a connected component
		is 		$\frac{n}{2^{i+2}}\cdot\frac{2^{i}}{\ell}=\frac{1}{4}\cdot\frac{n}{\ell}$ (since at most $\frac{n}{2^{i+2}}$ of the $\frac{n}{2^{i+1}}$ intervals can be connected, and each has $\frac{2^i}{\ell}$ points).
		In particular, $B_{{\cal J},H}^{+}\setminus B_{{\cal J}}$ must contain
		at least $\frac{n}{4\ell}$ points. As $H$ is an oblivous $\nu$-reliable spanner, it follows that 
		\[
		(1+\nu)\cdot|B_{{\cal J}}|\ge\mathbb{E}_{H\sim{\cal D}}\left[|B_{{\cal J},H}^{+}|\right]\ge|B_{{\cal J}}|+\frac{n}{4\ell}\cdot\Pr\left[\overline{\Psi_{{\cal J}}}\right]~.
		\]
		Hence $\Pr\left[\overline{\Psi_{{\cal J}}}\right]\le\nu\cdot|B_{{\cal J}}|\cdot\frac{4\ell}{n}<\frac{1}{2}$, and thus $\mathbb{E}_{H\sim{\cal D}}\left[\left|E(G_{H}^{{\cal J}})\right|\right]\ge\frac{n}{2^{i+2}}\cdot\Pr\left[\Psi_{{\cal J}}\right]\ge\frac{n}{2^{i+3}}$.
		
		We will abuse notation and state $(k,j)\in{\cal J}$ if the $k$'th
		index in ${\cal J}$ is $j$ (i.e. $j_{k}=j$). Next, we sample ${\cal J}$
		uniformly at random for all the possible $\frac{n}{2^{i+1}}$-tuples,
		and thus $\Pr\left[(k,j)\in{\cal J}\right]=\frac1\ell$. It holds that for every
		$H\in\supp({\cal D})$,
		\begin{align*}
			\mathbb{E}_{{\cal J}}\left[\left|E(G_{H}^{{\cal J}})\right|\right] & =\sum_{\left(A_{j}^{k},A_{j'}^{k'}\right)\in E(G_{H})}\Pr\left[(k,j),(k',j')\in{\cal J}\right]=\frac{1}{\ell^2}\cdot|E(G_{H})|~.
		\end{align*}
		We now sample both a subgraph $H\sim\cD$, and independently a
		tuple ${\cal J}$. It holds that: 
		\[
		\frac{1}{\ell^{2}}\cdot\mathbb{E}_{H}\left[\left|E(G_{H})\right|\right]=\mathbb{E}_{H}\left[\mathbb{E}_{{\cal J}}\left[\left|E(G_{H}^{{\cal J}})\right|\right]\right]=\mathbb{E}_{{\cal J}}\left[\mathbb{E}_{H}\left[\left|E(G_{H}^{{\cal J}})\right|\right]\right]\ge\mathbb{E}_{{\cal J}}\left[\frac{n}{2^{i+3}}\right]=\frac{n}{2^{i+3}}~,
		\]
		and thus $\mathbb{E}_{H}\left[\left|E(G_{H})\right|\right]\ge n\cdot\frac{\ell^{2}}{2^{i+3}}=\Omega(\frac{n}{2^{i}\cdot\nu^{2}})$
		as required.
	\end{proof}
 
	Consider a pair $p<q\in [n]$ such that $2^{w}\le q-p<2^{w+1}$.
	The event $(p,q)\in H$ occurs if and only if all the events $\left\{ (p,q)\in{\cal E}_{H}^{i}\right\} _{i=0}^{w}$
	occurred (note that all these $w+1$ events are actually equivalent).
	As $q-p\ge2^{w}>\sum_{i=0}^{w}2^{i-1}$, it holds that 
	\begin{align*}
		\Pr\left[(p,q)\in H\right]\cdot(q-p) & \ge\sum_{i=0}^{w}\Pr_{H\sim{\cal D}}\left[(p,q)\in{\cal E}_{H}^{i}\right]\cdot2^{i-1}\\
		& =\sum_{i=0}^{m-1}\Pr_{H\sim{\cal D}}\left[(p,q)\in{\cal E}_{H}^{i}\right]\cdot2^{i-1}\ge\sum_{i=s+3}^{m-2}\Pr_{H\sim{\cal D}}\left[(p,q)\in{\cal E}_{H}^{i}\right]\cdot2^{i-1}~,
	\end{align*}
	where the equality holds as for every $i\ge w+1$, $\Pr_{H\sim{\cal D}}\left[(p,q)\in{\cal E}_{H}^{i}\right]=0$.
	By \Cref{clm:PathLBSingleScale}
	\begin{align*}
		\mathbb{E}_{H\sim{\cal D}}\left[w(H)\right] & =\sum_{p<q}\Pr_{H\sim{\cal D}}\left[(p,q)\in H\right]\cdot(q-p)\\
		& \ge\sum_{p<q}\sum_{i=s+3}^{m-2}\Pr_{H\sim{\cal D}}\left[(p,q)\in{\cal E}_{H}^{i}\right]\cdot2^{i-1}\\
		& =\sum_{i=s+3}^{m-2}\mathbb{E}_{H\sim{\cal D}}\left[\left|{\cal E}_{H}^{i}\right|\right]\cdot2^{i-1}\\
		& \ge\sum_{i=s+3}^{m-2}\Omega(\frac{n}{2^{i}\cdot\nu^{2}})\cdot2^{i-1}\\
		& =\frac{n}{\nu^{2}}\cdot\Omega(m-s-4)=\frac{n}{\nu^{2}}\cdot\Omega(\log(n\cdot\nu))~,
	\end{align*}
	where the last equality holds as $m=\log n$, $s=\log\frac{1}{\nu}$, and thus $m-s-4=\log\frac{n\cdot\nu}{16}$. The theorem now follows.

\end{proof}

\bibliographystyle{alpha}
\bibliography{reliable,LSObib}

\appendix
\section{A Helpful Lemma}\label{app:sterling}

\begin{lemma}\label{lem:sterling}
For any integers $n,k$ and $0<\alpha<1$ such that $\alpha n\ge k>0$ is an integer too, we have that
\[
\frac{{\alpha n\choose k}}{{n\choose k}} \le O(\sqrt{k})\cdot \alpha^k~.
\]
\end{lemma}

\begin{proof}
We may assume that $k<\alpha n$, since when $k=\alpha n$ it suffices to prove that $\frac{1}{{n\choose\alpha n}}\le \alpha^{\alpha n}$, which holds by the estimate ${n\choose\alpha n}\ge \left(\frac{n}{\alpha n}\right)^{\alpha n}$.

Recall that Sterling approximation asserts that
\[
m!=\Theta(\sqrt{m})\cdot\left(\frac me\right)^m~.
\]
Thus we have that
\begin{eqnarray*}
\frac{{\alpha n\choose k}}{{n\choose k}} &=& 
\frac{\frac{(\alpha n)!}{k!\cdot(\alpha n-k)!}}{\frac{n!}{k!\cdot(n-k)!}}
=
\frac{(\alpha n)!\cdot (n-k)!}{(\alpha n-k)!\cdot n!}\\
&\le&O(1)\cdot\frac{\sqrt{\alpha n}(\alpha n/e)^{\alpha n}\cdot\sqrt{n-k}((n-k)/e)^{n-k}}{\sqrt{\alpha n-k} ((\alpha n-k)/e)^{\alpha n-k}\cdot\sqrt{n}(n/e)^{n}}\\
\end{eqnarray*}
Note that $\frac{\sqrt{\alpha n}}{\sqrt{\alpha n-k}}\le 2\sqrt{k}$ (using that $0<k<\alpha n$), and that $\frac{\sqrt{n-k}}{\sqrt{n}}\le 1$, we also replace $\alpha n-k$ in the denominator by $\alpha n$, and $n-k$ in the numerator by $n$, so the bound we get is
\[
O(\sqrt{k})\cdot\frac{(\alpha n)^{\alpha n}\cdot n^{n-k}}{ (\alpha n)^{\alpha n-k}\cdot n^n}=O(\sqrt{k})\cdot\alpha^k~.
\]

\end{proof}

\section{Light Reliable $O(\log n)$-Spanner}\label{sec:logNstretchSpanner}
For stretch $t=\log n$, the lightness of \Cref{thm:general} is $\nu^{-2}\cdot\polylog(n)$, while by \Cref{thm:PathObliviousLB}, $\Omega(\nu^{-2}\cdot\log n)$ lightness is necessary (even for preserving only the connectivity of the path metric). It is interesting to understand what is the best poly-logarithmic factor we can get (as it becomes the major factor).
The PPCS of Filtser and Le \cite{FL22} (\Cref{lem:pairwise-partition-general}) is constructed to optimize the stretch ($\rho$ parameter). 
Alternately, one can use the sparse covers of Awerbuch and Peleg \cite{AP90} to construct better PPCS for the $O(\log n)$-stretch regime. 
Given $n$-point metric space and parameter $\Delta>0$, Awerbuch and Peleg \cite{AP90} constructed 
a collection of $O(\log n)$ $\Delta$-bounded partitions, such that every ball of radius $\Omega(\frac{\Delta}{\log n})$ is fully contained in some cluster, in one of the partitions. One can verify that this exact collection of partitions is in particular 
$(O(\log n), O(\log n),\Omega(\frac{1}{\log n}))$-PPCS (this is the same argument as in \Cref{lem:pairwise-partitionlarge-ddim}).
By applying \Cref{thm:pairwise_partition_cover_to_ultrametric_cover} with $k=O(\log^2n)$, we obtain a $O(\log^3 n)$-light $\left(\tilde{O}(\log^2n),O(\log n)\right)$-$O(\log^2n)$-HST cover.
Next, by applying \Cref{thm:ultrametric_cover_to_reliable_spanner} we obtain an oblivious $\nu$-reliable $O(\log n)$-spanner with lightness $\tilde{O}(\nu^{-2}\cdot\log^9n)$.

As we show next, it is possible to obtain a quadratic improvement in the dependence on $n$ by constructing a light ultrametric cover directly using stochastic tree embeddings.
\begin{lemma}\label{lem:HST-Cover-FRT}
	\sloppy Every $n$-point metric space $(X,d_X)$ admits an $(O(\log n),O(\log n))$-$2$-HST cover, where the sum of weights of all the HST's in the cover is $O(\log^2n)$ times the MST of $X$.
\end{lemma}
\begin{proof}[Proof sketch.]
	Given an $n$-point metric space $(X,d_X)$, Fakcharoenphol, Rao, and Talwal \cite{FRT04} (see also \cite{Bartal04}) constructed a stochastic embedding into $2$-HST's with expected distortion $O(\log n)$. Specifically, they constructed a distribution $\cD$ over dominating $2$-HST's, such that for every $x,y\in X$, $\E_{T\sim\cD}\left[d_T(x,y)\right]\le O(\log n)\cdot d_X(x,y)$.
	In particular, by Markov, it holds that $\Pr_{T\sim{\cal D}}\left[d_T(x,y)\le O(\log n)\cdot d_X(x,y)\right]\ge \frac12$.
	
	\sloppy Sample $O(\log n)$ $2$-HST's $T_1,T_2,\dots$ from the distribution of \cite{FRT04}. By a standard application of Chernoff and union bounds, it holds that w.h.p., for every $x,y\in X$, $\min_i d_{T_i}(x,y)\le  O(\log n)\cdot d_X(x,y)$.
	Let $M$ be the MST of $X$, consisting of the edges $\{(x_j,y_j)\}_{j=1}^{n-1}$. As every pair has expected distortion $O(\log n)$, by linearity of expectation, it holds that $\E_{T\sim\cD}\left[\sum_{j=1}^{n-1} d_T(x_j,y_j)\right]\le O(\log n)\cdot \sum_{j=1}^{n-1} d_X(x_j,y_j)=O(\log n)\cdot w(M)$. It follows that the expected weight of $T_i$, for every $i$, is at most $O(\log n)\cdot w(M)$.
	The lemma follows as we can repeat the sampling until we get a collection satisfying all point pairs of total weight  $O(\log^2n)$ times the MST.
\end{proof}

Finally, we can apply \Cref{thm:ultrametric_cover_to_reliable_spanner} on the HST cover from \Cref{lem:HST-Cover-FRT}. Observe that for the proof of \Cref{thm:ultrametric_cover_to_reliable_spanner} it actually was enough that the total weight of all the HST's in the cover is bounded by $\psi\cdot \tau$ (rather that the weight of each HST is bounded by $\psi$). 
\begin{corollary}\label{cor:HSTcoverGeneralMetricLogn}
	For any parameters $\nu \in (0, 1/6)$, any metric space admits an oblivious $\nu$-reliable $O(\log n)$-spanner with size $n\cdot \tilde{O}\left(\nu^{-2}\cdot\log^{3}n\right)$ and lightness $\tilde{O}(\nu^{-2}\cdot\log^4 n)$.
\end{corollary}

\end{document}